\documentclass[11pt]{article}

\usepackage[margin=1in]{geometry}
\usepackage[T1]{fontenc}
\usepackage{lmodern}          
\usepackage{hyperref}
\usepackage{url}
\usepackage{float}
\usepackage{amsthm}
\usepackage{multirow}
\usepackage{algorithm}
\usepackage[noend]{algpseudocode}

\theoremstyle{proposition}

\theoremstyle{lemma}
\newtheorem{lemma}{Lemma}

\theoremstyle{corollary}
\newtheorem{corollary}{Corollary}

\theoremstyle{theorem}
\newtheorem{theorem}{Theorem}

\theoremstyle{definition}

\theoremstyle{assumption}

\theoremstyle{remark}

\usepackage{amsmath,amssymb}
\usepackage{graphicx}
\usepackage{booktabs}
\usepackage[numbers]{natbib}  
\usepackage{authblk}          

\title{Improving score-based sampling via affine post-processing}

\author[1]{Curtis McDonald}
\author[2]{Kevin Rojas}
\author[2]{Molei Tao}
\author[1,3]{Peter Bartlett}
\affil[1]{Simons Institute, University of California Berkeley}
\affil[2]{Georgia Tech}
\affil[3]{Google Deepmind}

\date{Sept 30, 2026}  

\begin{document}
\maketitle

\begin{abstract}
Sampling from a probability distribution with access only to its log-density is a foundational problem in statistics, machine learning, and the physical sciences. This problem is studied in the data-free score-based diffusion literature, which samples the target density via estimating scores of densities produced by convolution with Gaussian random variables and implementing a reverse SDE. In this framework, score estimators are constructed via Markov chain Monte Carlo, importance sampling, or rejection sampling. The estimators often operate far beyond the regime where mixing time and error control can be established and have large bias and variance when estimating the underlying score. In our method we propose post-processing: splitting a fixed computation budget across multiple evaluations, querying the estimator at different points, and fitting a single pooled estimate via linear combination to minimize bias and variance at a target location. The final value is an affine map of a vector of multiple score estimates. We show this reduces score estimation error when compared to the computation-matched base estimator queried at a single point. We demonstrate improved sample quality over other score estimation methods matched for computation cost on difficult target densities such as log-concave densities with large condition number, multimodal densities, and high-dimensional densities.\footnote{Code available at \url{https://github.com/KevinRojas1499/Score-Affine-Post-Processing}}
\end{abstract}

\section{Introduction}

Sampling from a probability density $\pi(x) \propto e^{-f(x)}$ given only the negative log-density $f$ is a foundational problem across Bayesian inference, statistical physics, and computational biology, among many other fields. Traditionally, this problem has been approached with sampling algorithms such as time-invariant Markov chain Monte Carlo, annealing methods, and tempering algorithms. More recently, the field of data-free score-based diffusion has offered an alternative method: sampling from the target density by reversing a noising process of Gaussian convolution. Implementing such a sampling algorithm requires estimating the noised score function with only knowledge of $f$.

Consider a random variable $U$ defined by a linear combination of a sample from $\pi$ and a Gaussian random variable,
\begin{align}
U = a X+\sqrt{1-a^{2}}Z, a \in [0,1], X \sim \pi, Z \sim N(0,I).\label{U_def}
\end{align}
Let $p_{a}$ denote the density of $U$ at signal level $a$, and write $s(u,a) = \nabla_{u}\log p_{a}(u)$ for its score.
The ability to estimate $s(u,a)$ for various inputs $u$ and $a$ is what facilitates sampling from $\pi$ via score-based diffusion \citep{song2021scorebased}.

The score is equal to an expectation over the reverse conditional density $p_{a}(x|u)$ using Tweedie's formula \citep{robbins1992empirical}. This has motivated a line of diffusion-based Monte Carlo methods in which the score of the noised target is estimated pointwise by a random estimator $\hat{s}(u,a)$ computed with knowledge of $f$ alone. This estimator should be considered a random variable with internal randomness: when queried at a fixed $u$ and $a$ value, it will return 
a value with some probability distribution. With knowledge of $f$, this estimator can be a form of importance sampling (IS) \citep{mcdonald2022proposal}, rejection sampling (RS) \citep{he2024zeroth}, or Markov chain Monte Carlo (MCMC) \citep{huang2024reverse, huang2024faster, grenioux2024stochastic} approximating an empirical average of Tweedie's formula to estimate the score.

The reverse conditional density $p_{a}(x|u)$ which defines the expectation in Tweedie's formula is often high-dimensional and multimodal, making it particularly difficult to estimate for sampling-based score estimators. Although such estimators may draw hundreds or thousands of proposal points, often only a handful contribute meaningfully to the empirical average. The effective sample size is therefore much smaller than the number of proposal points, and it grows little as more samples are drawn. As a result, these estimators suffer from high bias and variance at practical sample sizes, with significant error reduction only at sample sizes exponentially large in the state space dimension.
 
In this paper, we consider an alternative route to reducing bias and variance: querying the estimator at a fixed location $u$ across multiple signal levels $\{a_{i}\}_{i=0}^{N}$, and pooling the resulting estimates $\{\hat{s}(u,a_{i})\}_{i=0}^{N}$ together with the known terminal values $s(u,0) = -u$ and $s(u,1) = -\nabla f(u)$. For any vector $h \in \mathbb{R}^{N+1},g \in \mathbb{R}^{2},\nu \in \mathbb{R}$ we define an \textit{affine post-processor} as any estimation rule
\begin{align}
\bar{s}(u,a_{0}) = \sum\nolimits_{i=0}^{N} h_{i} \hat{s}(u,a_{i}) + g_{0}s(u,0) + g_{1}s(u,1)-\nu \label{affine_eqn}.
\end{align}
It is a form of linear smoother or ridge regression performing bias and variance reduction via a shrinkage procedure. We refer to the signal level $a_{0}$ as the on-policy level, and the queries $\hat{s}(u,a_{i}),i \neq 0$ as off-policy evaluations.

The choice of parameters $h = e_{0},~g = 0,~\nu = 0$, where $e_{0} = (1,0,\dots,0) \in \mathbb{R}^{N+1}$, we call ``pure on-policy estimation'' at $\hat{s}(u,a_{0})$, so simply relying on the score estimator's on-policy estimate is already a form of affine post-processing. If there is an alternative choice of $(h,g,\nu)$ which strictly reduces the average squared bias and variance at a signal level, then it reduces the path-averaged score error that controls the sampling error bound, which we observe as a reduction in terminal Wasserstein distance ($W_{2}$). We demonstrate that:
\begin{enumerate}\setlength{\itemsep}{2pt}\setlength{\topsep}{2pt}\setlength{\parskip}{0pt}
    \item The optimal choice of smoothing parameters relies on knowledge of the bias and variance of the base estimator. It can be written down in closed form, which we call the Wiener oracle post-processor. If it is not equal to $(e_{0}, 0,0)$, then optimally chosen affine post-processing strictly reduces score estimation error.
    \item We show pure on-policy estimation is optimal only under a specific zero-gradient condition most estimators will violate. In particular, for any unbiased estimator with non-zero variance, pure on-policy estimation is not optimal among affine post-processors.
    \item We construct practical, non-oracle post-processors via ridge regression in the Legendre polynomial basis, and in our experiments their score error captures most of the improvement available to the Wiener oracle.
    \item Post-processing is compared at matched computation cost: a fixed MCMC  budget of $\mathcal{N}_{\text{eval}}$ evaluations is split across the $N+1$ query levels, and the pooled estimate is compared to a single on-policy estimate given the full $\mathcal{N}_{\text{eval}}$ budget. The effective sample size grows slowly in $\mathcal{N}_{\text{eval}}$, thus the on-policy estimator gains little from the extra budget, while the split-and-pool estimator reduces score estimation error substantially.
\end{enumerate}
Empirically, diffusion-based samplers built on pure on-policy score estimation degrade quickly as dimension grows. The Tweedie integral becomes hard to estimate directly in high dimensions, and base score estimates carry large bias and large variance. Post-processing is a general procedure that can be applied on top of any base score estimator. We run experiments with IS and Reverse Diffusion Monte Carlo (RDMC) \citep{huang2024reverse} as base estimators in high-dimensional problems. Post-processing reduces score estimation error substantially, and improves sample quality with $W_2$ distance up to two orders of magnitude smaller at equal compute. This makes diffusion-based samplers practical in dimensions where their base estimators would otherwise fail.

\section{Background}

\subsection{Diffusion Models}
Fix a horizon $T > 0$ and a decreasing signal level schedule $\alpha : [0,T] \to (0,1]$, with
$\alpha(0) = 1$. The forward stochastic process
\begin{align}
    dX_{t} = \frac{\dot{\alpha}(t)}{\alpha(t)}\,X_{t}\,dt
           + \sqrt{-\frac{2\dot{\alpha}(t)}{\alpha(t)}}\,dB_{t},
    \qquad X_{0} \sim \pi,
    \label{eq:forward}
\end{align}
has marginals $X_{t} \overset{d}{=} \alpha(t)X_{0} + \sqrt{1-\alpha(t)^{2}}\,Z$
with $Z \sim N(0,I)$ independent of $X_{0}$, so $X_{t} \sim p_{\alpha(t)}$;
it is the variance-preserving Ornstein--Uhlenbeck process run on the clock
$\alpha$. By the time-reversal of diffusions
\citep{song2021scorebased, kingma2021variational}, the reverse process
\begin{align}
    dY_{t} = -\frac{\dot{\alpha}(T-t)}{\alpha(T-t)}
             \Big(Y_{t} + 2\,s\big(Y_{t},\alpha(T-t)\big)\Big)dt
           + \sqrt{-\frac{2\dot{\alpha}(T-t)}{\alpha(T-t)}}\,dB_{t},
    \qquad Y_{0} \sim p_{\alpha(T)},
    \label{eq:reverse}
\end{align}
satisfies $Y_{t} \sim p_{\alpha(T-t)}$, so that $Y_{T} \sim \pi$. Different $\alpha$ time
schedules are reparametrizations of the same reverse dynamics; the schedule
decides only which signal levels a discretization visits. We note there are additional small errors introduced by initializing at $Y_{0} \sim N(0,I) \neq p_{\alpha(T)}$ and terminating at time $t'<T$ which are considered fully in the discrete process error analysis theory \citep{song2021maximum, lee2022convergence, chen2023sampling, chen2023improved, benton2024nearly, conforti2025kl, wang2024evaluating}.

In this work we use the linear schedule $\alpha(t) = 1-t$, for which
$-\dot{\alpha}/\alpha = 1/\alpha$. Specializing to the natural parameter schedule of the reverse process, the reverse process runs over signal levels $a$ directly, with $a$ increasing and $a \to 1$ approaching a sample from the target density
\begin{align}
    dY_{a} = \frac{1}{a}\big(Y_{a} + 2\,s(Y_{a},a)\big)\,da
           + \sqrt{\frac{2}{a}}\,dW_{a},
    \label{eq:reverse-linear}
\end{align}
Given a score estimator $\hat{s}$ and a grid
$0<a_{0} < a_{1} < \cdots < a_{M}<1$ with $h_{k} = a_{k+1} - a_{k}$, we initialize
$y_{0} \sim N(0,I) \approx p_{a_{0}}$ and take Euler--Maruyama steps
\begin{align}
    y_{k+1} = y_{k} + \frac{h_{k}}{a_{k}}\big(y_{k} + 2\,\hat{s}(y_{k},a_{k})\big)
            + \sqrt{\frac{2h_{k}}{a_{k}}}\,\xi_{k},
    \qquad \xi_{k} \sim N(0,I).
    \label{eq:em}
\end{align}
We use a non-uniform grid which follows the exponential-decay design of
\citet{chen2023improved,benton2024nearly,he2024zeroth}, which was shown to be optimal and prevents numerical instability in \eqref{eq:reverse-linear}, see Appendix~\ref{app:grid} for a precise definition. 

\subsection{Monte Carlo Score Estimators}

With knowledge of $f$, the score can be written as a posterior expectation using Tweedie's formula \citep{robbins1992empirical},
\begin{align}
    s(u,a) = E\Big[-\frac{u-aX}{1-a^{2}}\,\Big|\,U=u\Big] = \int -\left(\frac{u-ax}{1-a^{2}}\right) \frac{e^{-f(x)-\frac{1}{2(1-a^{2})}\|u - a x\|^{2}}}{\int e^{-f(x')-\frac{1}{2(1-a^{2})}\|u - a x'\|^{2}}dx' }dx ,\label{Tweedie_eqn}
\end{align}
where the expectation is over the reverse conditional density $p_{a}(x|u) \propto e^{-f(x)}e^{-\|u-ax\|^{2}/2(1-a^{2})}$. Every base score estimator we consider replaces this expectation with an empirical average of a set of sampled proposal points. Each estimator generates proposal points $\{z_{i}\}_{i=1}^{n}$ and non-negative weights $\{\omega_{i}\}_{i=1}^{n}$ summing to 1, and estimates the score as the average $\hat{s}(u,a) = -\frac{u}{1-a^{2}}+\frac{a}{1-a^{2}}\sum_{i=1}^{n}\omega_{i}z_{i}$. The estimators differ only in how the $(z_{i},\omega_{i})$ are produced.
IS draws $z_{i}$ i.i.d.\ from the Gaussian proposal $N(u/a, \tfrac{1-a^{2}}{a^{2}}I)$ and sets $\omega_{i} \propto e^{-f(z_{i})}$. 
RS \citep{he2024zeroth} draws from the same proposal, accepts each point with probability proportional to $e^{-f(z_{i})}$, and places equal weight on the accepted points, which are exact draws from $p_{a}(x|u)$.
RDMC \citep{huang2024reverse} runs an unadjusted Langevin chain of length $K$ on $p_{a}(x|u)$, using $\nabla f$, and takes $z_{i}$ to be the chain iterates with equal weights; unlike IS and RS the $z_{i}$ are correlated, so the estimator is a time average along one trajectory.

\subsection{Ridge Regression}

Fix a location $u$ and a coordinate $j$, which we suppress. The score $s(u,a) = s_{j}(u,a)$ is a smooth function of $a \in [0,1]$, and the base estimator returns noisy, possibly biased observations $\hat{\mathbf{s}}(u) = (\hat{s}(u,a_{0}),\dots,\hat{s}(u,a_{N}))$ of that function at the design levels $a_{0},\dots,a_{N}$. Its two endpoint values are known exactly, $v(u) = (s(u,0),\,s(u,1)) = (-u,\,-\partial f(u))$. Estimating a function at a target point from noisy evaluations elsewhere is a classical problem, and there are two equivalent expressions.

Choose weights $h = h(a_{0}) \in \mathbb{R}^{N+1}$ and $g = g(a_{0}) \in \mathbb{R}^{2}$ and return the linear combination
\begin{align}
    \bar{s}(u,a_{0}) = \langle h, \hat{\mathbf{s}}(u)\rangle + \langle g, v(u)\rangle. \label{eq:smoother_view}
\end{align}
This is the affine post-processor \eqref{affine_eqn} with $\nu = 0$. The weights depend on $a_{0}$ and the design levels but not on $u$ or on the observed values.

Fix an orthogonal functional basis $\{\phi_{k}\}_{k\geq 0}$ of $L^{2}[0,1]$ (we use shifted Legendre polynomials) and a degree $K \leq N$. Fit coefficients $\beta \in \mathbb{R}^{K+1}$ by weighted ridge regression of the observations onto the basis, constrained to pass through the endpoints, and evaluate the fitted polynomial at $a_{0}$:
\begin{align}
    \bar{s}(u,a_{0}) = \sum_{k=0}^{K}\beta^{*}_{k}(u)\,\phi_{k}(a_{0}) = \langle b, \beta^{*}(u)\rangle, \qquad b = (\phi_{0}(a_{0}),\dots,\phi_{K}(a_{0})). \label{eq:regression_view}
\end{align}
Writing $\Phi_{i,k} = \phi_{k}(a_{i})$ for the design matrix and $\Psi_{e,k} = \phi_{k}(e)$, $e \in \{0,1\}$, for the basis at the endpoints, the ridge solution is linear in the data: $\beta^{*}(u) = M_{1}\hat{\mathbf{s}}(u) + M_{2}v(u)$ for matrices $M_{1}, M_{2}$ that depend only on $\Phi$, $\Psi$, and the penalties. Substituting into \eqref{eq:regression_view},
\begin{align}
    \bar{s}(u,a_{0}) = \langle M_{1}^{\top}b,\ \hat{\mathbf{s}}(u)\rangle + \langle M_{2}^{\top}b,\ v(u)\rangle,
\end{align}
which is \eqref{eq:smoother_view} with $h = M_{1}^{\top}b$ and $g = M_{2}^{\top}b$. Every ridge fit is therefore a linear smoother. Theorem~\ref{ridge_adjoint} gives the converse direction and the explicit correspondence: the smoother weights $(h,g)$ solve their own quadratic program, in which the ridge penalty on $\beta$ becomes a penalty on how far $(h,g)$ is from reproducing each basis function at $a_{0}$.

Estimating a function at a target point by a linear combination of noisy evaluations elsewhere is classical statistics: it is the Wiener--Kolmogorov filtering problem \citep{wiener1949extrapolation,kolmogorov1941interpolation}, kriging in geostatistics \citep{matheron1963principles,williams2006gaussian}, and the theory of linear smoothers in nonparametric regression \citep{silverman1984spline,buja1989linear}. The equivalence between fitting a penalized regression and optimizing the weights directly is made explicit in the optimal-recovery literature \citep{donoho1994statistical} and in balancing-weight estimators from causal inference \citep{zubizarreta2015stable,bruns2026augmented}. Our contribution is to specialize this machinery to combine multiple score estimates with an affine map.

\section{Affine Post-Processing Theory}

The performance of the base score estimator is captured by its coordinate-wise bias and variance functions $B$ and $V$. Note the estimator $\hat{s}$ has internal randomness, and the input value $U$ is itself a random variable. Throughout, $E[\cdot]$ and $Var(\cdot)$ without subscripts denote expectation and variance over this internal randomness with $(u,a)$ held fixed; averages over the population are always written with an explicit subscript, $E_{U\sim p_{a}}[\cdot]$. At a signal level $a$ and coordinate $j$, we refer to the population average squared bias and the population average variance over $U \sim p_{a}$ as $\tilde{B}$ and $\tilde{V}$:
\begin{align}
B_{j}(u,a) &= E[\hat{s}_{j}(u,a)] - s_{j}(u,a)&V_{j}(u,a) = \text{Var}(\hat{s}_{j}(u,a))\\
\tilde{B}_{j}^{2}(a) &= E_{U \sim p_{a}}[B_{j}^{2}(U,a)]&\tilde{V}_{j}(a) = E_{U \sim p_{a}}[V_{j}(U,a)].
\end{align}

In a score-based diffusion, which is a discrete approximation to a continuous SDE, the terminal KL error between the target $\pi$ and the density produced by sampling can be decomposed into three contributions: discrete step size error, initialization error, and path averaged score error \citep{song2021maximum, lee2022convergence, chen2023sampling, chen2023improved, benton2024nearly, conforti2025kl, wang2024evaluating}. The third source of error is a property of the quality of the score estimator and can be expressed for some weight function $\omega$ as the average of population-wide squared bias and variance,
\begin{align}
\int_{0}^{1}\omega(a)E_{U \sim p_{a}}[E[\|\hat{s}(u,a) - s(u,a)\|^{2}]]da = \sum_{j}\int_{0}^{1}\omega(a)(\tilde{B}_{j}^{2}(a)+\tilde{V}_{j}(a))da.
\end{align}
The score error decomposes by coordinate $j$, and our affine post-processor fits each coordinate $j$ independently. As such, we now suppress the index $j$ from our notation and consider this method as repeated in each coordinate. We define $L^{2}$ error of an affine post-processing vector as the functional
\begin{align}
J(h,g,\nu, a_{0}) = E_{U \sim p_{a}}[E[(\bar{s}(u,a_{0})-s(u,a_{0}))^{2}]]\label{L2}.
\end{align}
\subsection{Oracle Theory: The Wiener Post-Processor}
\label{sec:oracle}

Fix an on-policy signal level $a_{0}$ and denote the on-policy $U$-dependent random variables $B(u) = B(u,a_{0})$, $s(u) = s(u,a_{0})$. Since the terminal scores are observations with zero bias and variance, we stack them into a larger observation vector $\hat{s} = (\hat{s}(u,a_{i})_{i=0}^{N}, s(u,0), s(u,1))\in \mathbb{R}^{N+3}$ with weights $\tilde h = (h_{0},\dots,h_{N},g_{0},g_{1}) \in \mathbb{R}^{N+3}$. Define the random variable $\tilde m(u) \in \mathbb{R}^{N+3}$ and population average on- and off-policy variance $\mathcal{V} \in \mathbb{R}^{(N+3)\times(N+3)}$ as
\begin{align}
\tilde m_{i}(u) &= E[\hat{s}(u,a_{i})] = s(u,a_{i}) + B(u,a_{i}), & \mathcal{V}_{ii} &= E_{U \sim p_{a_{0}}}[V(U,a_{i})], \quad i = 0,\dots,N\\
\tilde m_{N+1}(u) &= s(u,0),\quad \tilde m_{N+2}(u) = s(u,1), & \mathcal{V}_{ii} &= 0,  \quad\quad\quad\quad\quad i = N+1,N+2.
\end{align}
Then the affine post-processing estimator can be expressed as $\bar{s}(u,a_{0}) = \langle \tilde h, \hat{s}(u)\rangle - \nu$ with mean $\langle \tilde h, \tilde m(u) \rangle - \nu$, and we write $J(\tilde h,\nu,a_{0})$.
Since $\tilde{h},\nu$ are signal-level dependent but not $U$ dependent, they factor out of $U$ expectations. As such, $J$ can be expressed as a convex quadratic in $\tilde{h},\nu$ with a single closed form minimizer called the Wiener oracle.

\begin{theorem}\label{wiener_oracle}
Writing $z = (\tilde h, \nu) \in \mathbb{R}^{N+4}$ and $\tilde m = \tilde m(U), s = s(U)$, $J$ is the explicit convex quadratic
\begin{small}
\begin{align}
J(z, a_{0}) &= z^{\top} \mathcal{M}(a_{0})\, z + \mu(a_{0})^{\top} z + E_{U\sim p_{a_{0}}}[s^{2}], \label{J_quad_form}\\
\mathcal{M} &=
\begin{pmatrix}
E_{U\sim p_{a_{0}}}[\tilde m \tilde m^{\top}] + \mathcal{V} & -E_{U\sim p_{a_{0}}}[\tilde m] \\
-E_{U\sim p_{a_{0}}}[\tilde m]^{\top} & 1
\end{pmatrix},
\quad
\mu =
\begin{pmatrix} -2\, \text{Cov}_{U\sim p_{a_{0}}}[\tilde m\, s] \\ 0 \end{pmatrix}.
\label{eq:quad}
\end{align}
\end{small}
The minimum of $J$ is always attained. When $\mathrm{Cov}_{U\sim p_{a_{0}}}(\tilde m) + \mathcal{V}$ is invertible the minimizer is unique and given by the Wiener estimator below; otherwise, replacing the inverse by the Moore--Penrose pseudo-inverse gives the minimum-norm minimizer.
\begin{align}
\tilde h^{*}(a_{0}) &= \big[ \mathrm{Cov}_{U\sim p_{a_{0}}}(\tilde m) + \mathcal{V} \big]^{-1} \mathrm{Cov}_{U\sim p_{a_{0}}}\big(\tilde m, s\big),
\quad\nu^{*}(a_{0}) = \tilde h^{*}(a_{0})^{\top}E_{U\sim p_{a_{0}}}[\tilde m]\label{eq:wiener}.
\end{align}
\end{theorem}

\begin{corollary}\label{off_policy_optimal}
With $\tilde m, s, B$ as $U$ dependent random variables as in Theorem~\ref{wiener_oracle}, let $z_{\mathrm{on}} = (e_{0}, 0)$, with $e_{0} = (1,0,\dots,0) \in \mathbb{R}^{N+3}$, denote the on-policy estimator. The gradient of $J$ at $z_{\mathrm{on}}$ is
\begin{align}
\partial_{\nu}J(z_{\mathrm{on}}, a_{0}) &= -2\,E_{U\sim p_{a_{0}}}[B] \quad\quad \partial_{\tilde h_{0}}J(z_{\mathrm{on}}, a_{0}) = 2\big(\text{Cov}_{U\sim p_{a_{0}}}[B\, s] + J(z_{\mathrm{on}}, a_{0})\big),\\
\partial_{\tilde h_{i}}J(z_{\mathrm{on}}, a_{0}) &= 2\,E_{U\sim p_{a_{0}}}[B\,\tilde m_{i}], \qquad i = 1,\dots,N+2&.
\end{align}
The on-policy estimator is the Wiener optimal solution if and only if
\begin{align}
E_{U\sim p_{a_{0}}}[B] &= 0, \qquad
\mathrm{Cov}_{U\sim p_{a_{0}}}(B, s) = -J(z_{\mathrm{on}}, a_{0}),\\
E_{U\sim p_{a_{0}}}[B\,\tilde m_{i}] &= 0, \qquad i = 1,\dots,N+2.
\end{align}
If any of these fails, some $z = (\tilde h,\nu)$, depending only on the signal level $a_{0}$, strictly reduces the population-average squared bias and variance relative to pure on-policy estimation.
\end{corollary}

\begin{corollary}\label{unbiased_wiener}
Assume the base estimator is unbiased, $E[\hat{s}(u,a)] = s(u,a)$ for all $(u,a)$. If $E_{U\sim p_{a_{0}}}[V(U,a_{0})] > 0$, then $z_{\mathrm{on}}$ is \textbf{not} the Wiener optimal solution: some affine post-processor has strictly smaller population-average squared bias and variance than pure on-policy estimation.
\end{corollary}

\subsection{Non-Oracle Post-Processors: Legendre Polynomial Regression}\label{sec:nonoracle}

The Wiener estimator cannot be computed in practice, since it requires knowledge of the score itself as well as other population average quantities which are not known. As such, other sub-optimal choices of $\tilde{h}$ must be implemented, and we set $\nu = 0$ throughout.

Fix an orthogonal functional basis $\{\phi_{k}\}_{k \geq 0}$ on $L^{2}[0,1]$ such as Legendre polynomials. For some finite degree $K$, define the design matrix $\Phi_{i,k} = \phi_{k}(a_{i})$ whose columns are the basis functions evaluated at the signal levels, and the rows are evaluations at the same signal level. We denote the top row (row 0) of the design matrix as the vector, $b = (\phi_{0}(a_{0}),\dots,\phi_{K}(a_{0}))$, the on-policy evaluations of the functional basis. Denote the terminal evaluations of the basis functions as $\Psi \in \mathbb{R}^{2\times(K+1)}$, $\Psi_{0,k} = \phi_{k}(0)$, $\Psi_{1,k} = \phi_{k}(1)$. Write $\hat{\mathbf{s}}(u) = (\hat{s}(u,a_{0}),\dots,\hat{s}(u,a_{N}))$ for the vector of score queries and $v(u) = (s(u,0),s(u,1))$ for the known terminal scores. For weights $(h,g)$, the defect vector $\delta(h,g) = \Phi^{\top}h + \Psi^{\top}g - b \in \mathbb{R}^{K+1}$ measures how far the post-processor is from reproducing each basis function exactly at $a_{0}$.

The affine post-processor can be constructed by fitting a set of coefficients in the functional basis and evaluating at the location $a_{0}$, or fitting the $(h,g)$ vector directly in an adjoint minimization problem.

\begin{theorem}\label{ridge_adjoint}
Assume $\Phi$ has full column rank ($N \geq K$) and $\Psi$ has full row rank. Let $\Sigma \in \mathbb{R}^{(N+1)\times(N+1)}$ be symmetric positive definite, and $\lambda > 0$. Let $\Omega$ be diagonal with $\Omega_{k,k} \in (0,\infty)$ for $k \notin \mathcal{K}$ and $\Omega_{k,k} = \infty$ for $k \in \mathcal{K}$, with the convention $1/\infty = 0$, so that the coefficients indexed by $\mathcal{K} \subset \{0,\dots,K\}$ are unpenalized. Consider ridge regression constrained to reproduce the terminal scores,
\begin{align}
\beta^{*}(u) = \operatorname*{arg\,min}_{\beta \in \mathbb{R}^{K+1}}\ \big(\hat{\mathbf{s}}(u) - \Phi\beta\big)^{\top}\Sigma^{-1}\big(\hat{\mathbf{s}}(u) - \Phi\beta\big) + \lambda\,\beta^{\top}\Omega^{-1}\beta \quad \text{s.t.}\quad \Psi\beta = v(u),
\end{align}
and its adjoint problem
\begin{align}
(h^{*},g^{*}) = \operatorname*{arg\,min}_{h\in\mathbb{R}^{N+1},\,g\in\mathbb{R}^{2}}\ \sum_{k\notin\mathcal{K}}\Omega_{k,k}\,\delta_{k}(h,g)^{2} + \lambda\,h^{\top}\Sigma h \quad \text{s.t.}\quad \delta_{k}(h,g) = 0,\ k\in\mathcal{K}. \label{adjoint_prob}
\end{align}
Both problems have unique solutions, and for every realization of $\hat{\mathbf{s}}(u)$,
\begin{align}
\bar{s}(u,a_{0}) = \sum_{k=0}^{K}\beta^{*}_{k}(u)\,\phi_{k}(a_{0}) = \sum_{i=0}^{N}h^{*}_{i}\,\hat{s}(u,a_{i}) + g^{*}_{0}\,s(u,0) + g^{*}_{1}\,s(u,1). \label{eq:affine_combination}
\end{align}
Writing $A = \Phi^{\top}\Sigma^{-1}\Phi + \lambda\,\Omega^{-1}$, $G = \Psi A^{-1}\Psi^{\top}$ and $P = I - A^{-1}\Psi^{\top}G^{-1}\Psi$, the solutions are
\begin{align}
\beta^{*}(u) &= P A^{-1}\Phi^{\top}\Sigma^{-1}\,\hat{\mathbf{s}}(u) + A^{-1}\Psi^{\top}G^{-1}\,v(u),\\
h^{*} &= \Sigma^{-1}\Phi A^{-1}P^{\top} b, \qquad
g^{*} = G^{-1}\Psi A^{-1} b. \label{eq:ridge_closed}
\end{align}
\end{theorem}
\begin{proof}
    See a further discussion of the ridge regression theory in Appendix \ref{lin_reg_theory}, and a full proof in Lemma \ref{dropped_penalty} which follows from the KKT conditions for constrained convex minimization problems.
\end{proof}

In practice we take $\Omega$ and $\Sigma$ diagonal, with $\mathcal{K} = \{0\}$ so the fit reproduces constants exactly, and $\Omega_{k,k} = k^{-3}$ for $k \geq 1$, so that $\lambda\sum_{k\geq 1}k^{3}\beta_{k}^{2}$ acts as a roughness penalty on the fitted polynomial. For $\Sigma$ we use either $\Sigma = I$, which costs nothing but assumes equal variance at every design level, or the diagonal estimate $\hat{\Sigma}_{i,i}(u) = \widehat{\mathrm{Var}}(\hat{s}(u,a_{i}))$ from repeated calls to the base estimator, which turns the fit into inverse-variance weighted regression at the cost of extra queries. With $V_{i} \in \{1, \hat{\Sigma}_{i,i}\}$ the adjoint problem \eqref{adjoint_prob} becomes
\begin{align}
\min_{h \in \mathbb{R}^{N+1},\, g \in \mathbb{R}^{2}}\ \sum_{k=1}^{K} \frac{1}{k^{3}}\Big(\sum_{i=0}^{N} h_{i}\phi_{k}(a_{i}) + g_{0}\phi_{k}(0) + g_{1}\phi_{k}(1) - \phi_{k}(a_{0})\Big)^{2} + \lambda\sum_{i=0}^{N} V_{i}h_{i}^{2} \label{eq:practical_qp}\\
\text{s.t.}\quad \sum_{i=0}^{N} h_{i}\phi_{k}(a_{i}) + g_{0}\phi_{k}(0) + g_{1}\phi_{k}(1) = \phi_{k}(a_{0}), \qquad k = 0. \nonumber
\end{align}
The solution $(h^{*},g^{*})$ depends on $a_{0}$, the design levels and $V$, but not on the observed values $\hat{\mathbf{s}}(u)$ or $v(u)$: it is a fixed rule applied to every score query at signal level $a_{0}$, which is clear from the adjoint form and hidden in the $\beta$ form. Appendix~\ref{lin_reg_theory} gives the full ridge-regression derivation.

\section{Experiments}

\label{sec:experiments-baselines}
We evaluate post-processing along three axes, each stressed by its own family of
targets (defined in Appendix~\ref{app:targets}). Section~\ref{sec:unimodal} tests
\emph{ill-conditioning}, on a Gaussian $\mathcal{N}(0,\Lambda_\kappa)$ with eigenvalues
log-spaced over $[1,\kappa]$, and \emph{geometry}, on a banana-shaped (twisted) Gaussian, a
zero-mean mixture of Gaussians with orthogonally oriented covariances, and an ill-conditioned
ellipsoidal shell. Section~\ref{sec:multimodal} tests
\emph{multimodality}, on a two-component anisotropic mixture under increasing separation,
an eight-mode mixture of Gaussians with unequal widths at the vertices of a rotated
hypercube, and a $40$-mode mixture of narrow isotropic Gaussians. In both sections we demonstrate improved dimension dependence.

\paragraph{Baselines and protocol.} We consider \emph{diffusion-based}
samplers including RDMC \citep{huang2024reverse} and 
IS as the score estimator, with and without post-processing. Due to the poor dimension dependence of the rejection sampling estimator ZODMC \citep{he2024zeroth} we skip this score estimator in our analysis. The samplers share the same $M=20$-step reverse loop in \eqref{eq:em}, in order to isolate the score estimation. We also test \emph{classical} samplers: unadjusted Langevin (LMC), no U-turn sampler (NUTS), and parallel tempering (PT). All samplers are compute-matched to the same number of gradient evaluations ($60$K per score query, $1.2$M per particle) and hyperparameters are detailed in Table~\ref{tab:protocol}. Unless stated otherwise, we report the Wasserstein distance ($W_2$)
averaged over $3$ seeds, with error bars for the minimum and maximum across seeds.

\begin{figure}
\centering
\includegraphics[width=\textwidth]{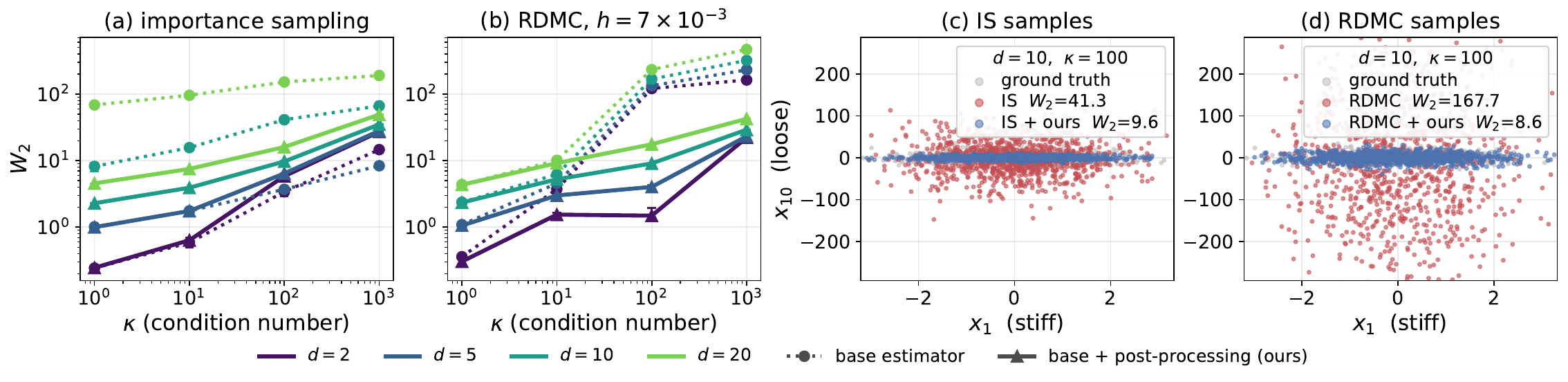}
\vspace{-\baselineskip}
\caption{(a),(b) Mean W2 for an Ill-conditioned Gaussian against the condition number $\kappa$. Post-processing improves dependency against dimension and $\kappa$. (c), (d) The base estimators generate scattered samples projected onto two dimensions, our approach recovers the true distribution.}
\label{fig:logconcave-grid}
\end{figure}
\subsection{Post-Processing improves score estimation for unimodal targets}\label{sec:unimodal}
Although diffusion-based samplers excel at sampling from low-dimensional multimodal targets \citep{he2024zeroth}, in high-dimensional unimodal situations their performance can be underwhelming. We demonstrate that post-processing can address this issue and that its advantages grow with dimension.

\textbf{Improved condition number dependence:}
On the ill-conditioned Gaussian we sweep both dimension and condition number. Figure~\ref{fig:logconcave-grid} reports the results of this experiment. For both IS and RDMC post-processing increases the accuracy of the generated samples, at high dimensions showing an order of magnitude reduction in $W_{2}$ and noticeably more concentrated samples to the ground truth density.

We recall the $L^{2}$ error function of an affine post-processor $J$ defined in \eqref{L2}. In Figure~\ref{fig:j-logconcave} we plot the average across dimension $J$ for the base estimator, the Wiener oracle post-processing (maximal signal that could be recovered) and compare it to our post-processed score. We find that our estimator can recover a significant portion of the Wiener oracle. Interestingly, different base estimators have distinct $J$ profiles depending on the time region where the score is hardest to estimate.

\begin{figure}
\vspace{-\baselineskip}
\centering
\includegraphics[width=\textwidth]{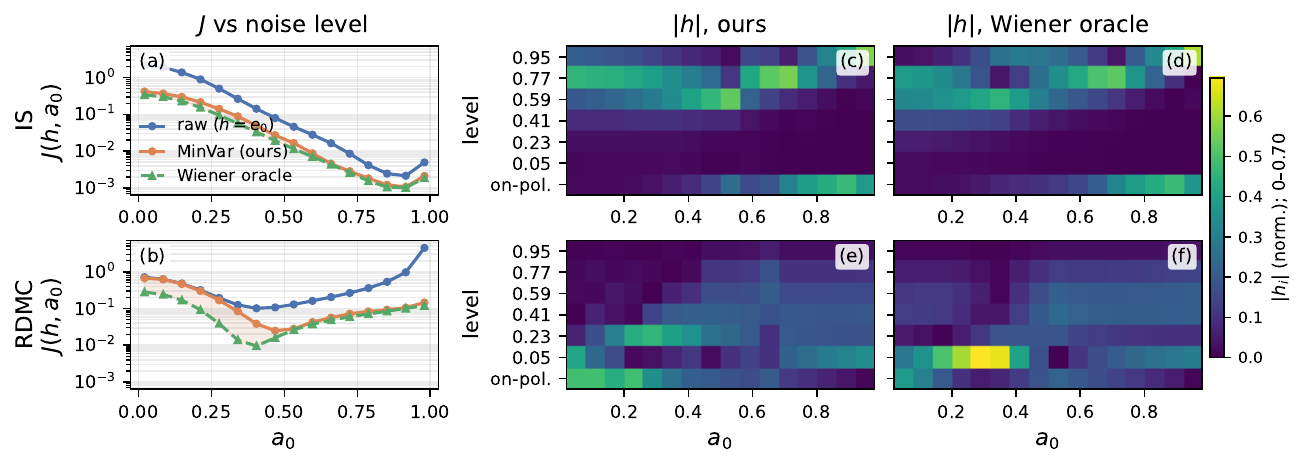}
\vspace{-\baselineskip}
\caption{Score error for a log-concave target at $d=10$, $\kappa=10$. \textbf{Left:} Average score error, for the base estimator, ours and the Wiener oracle. \textbf{Right:} Normalized weight of linear smoother vector $|h|$, each column is a vector, optimal allocations do not place all the weight on a single point.}
\label{fig:j-logconcave}
\end{figure}
\textbf{Improved geometry adherence:} We next test the three geometry targets, whose mass
concentrates on non-elliptical regions. Classical chains mix well on such targets, so we would expect them to outperform a diffusion sampler. The main comparison for these examples is diffusion sampling with and without post processing.

\begin{figure}
\centering
\includegraphics[width=\textwidth]{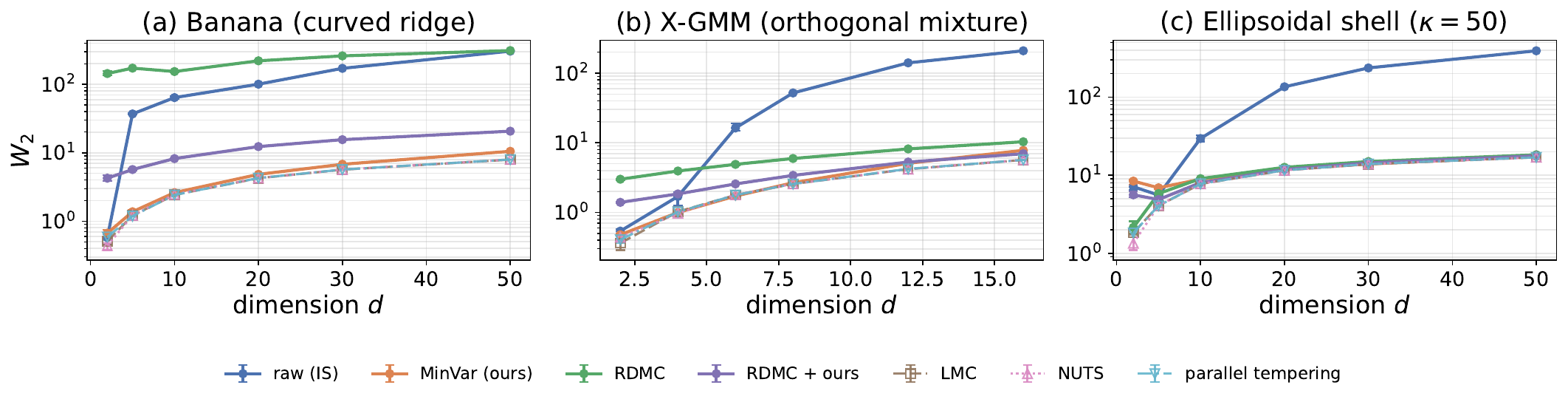}
\vspace{-\baselineskip}
\caption{Geometry targets, $W_2$ against dimension. The base diffusion samplers degrade with
dimension; post-processing closes the gap to the classical methods.}
\label{fig:geometry-dim-sweep}
\end{figure}

Figure~\ref{fig:geometry-dim-sweep} shows that the  IS base diffusion sampler fails on all three
targets, and fails more as the dimension grows. RDMC performs poorly on the Banana target and better on the other two. For instance, at the largest dimension, $W_{2}$ distance of IS is more than an order of magnitude worse than the classical methods on every target, 
and RDMC is as far off on the banana. Post-processing the same estimators closes this gap: at every dimension both samplers match the classical methods, with the relative improvements growing with the dimension.
This experiment demonstrates that although the base estimators may not explore the space well enough to estimate the score accurately, they 
carry enough information to \emph{construct} an accurate estimation when post-processed. 
Figure~\ref{fig:geometry-j} confirms this by looking at the score error itself on the
orthogonal mixture and the $40$-mode mixture of Section~\ref{sec:multimodal}:
post-processing lowers it at almost every noise level, and captures
most of the improvement available when compared to the Wiener oracle.

\begin{figure}
\centering
\includegraphics[width=\textwidth]{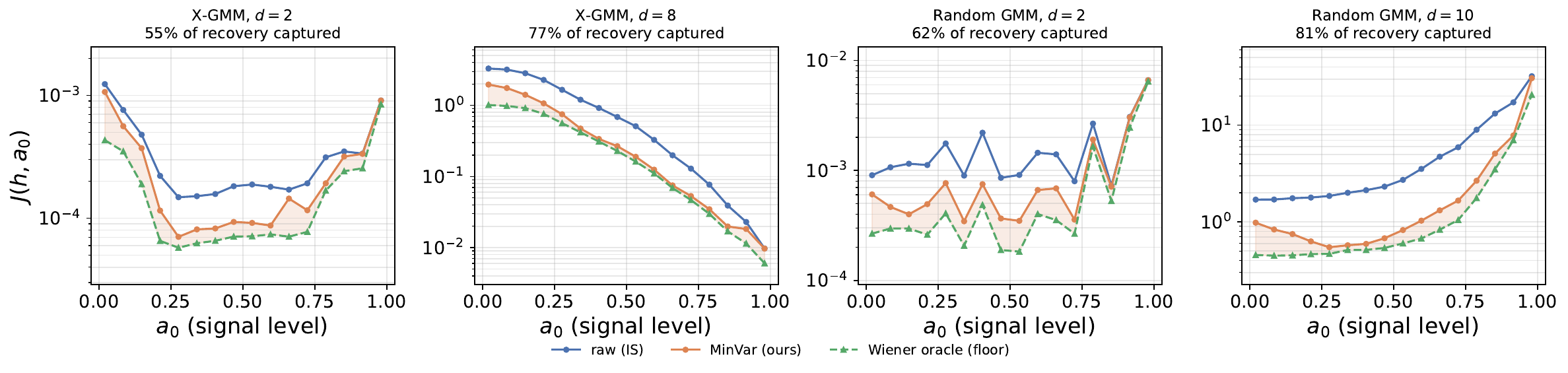}
\vspace{-\baselineskip}
\vspace{-\baselineskip}
\caption{\textit{Average score error.} When compared to the base estimator, post-processing reduces the score estimation error. The orange band denotes possible improvement given the evaluated samples.}
\label{fig:geometry-j}
\vspace{-\baselineskip}
\end{figure}

\subsection{Post-Processing improves score estimation for multimodal targets}\label{sec:multimodal}

\textbf{Improved dimension dependence on multimodal targest:} We consider two multimodal targets: a specifically chosen 8-mode GMM target with separated modes, and a 40-mode GMM with randomly generated modes. On the eight-mode  mixture we sweep $d$ from $10$ to $100$ using $1{,}000$ particles. We report $W_2$ on the $7$-dimensional span of their centers
to better capture the mode structure.

Figure~\ref{fig:vertex-dim} (a,b) shows the result. Raw importance sampling degrades quickly
with dimension, while post-processing the same estimator keeps its $W_2$ roughly constant
from $d=10$ to $d=100$, furthermore, plotting the first two PCA components of the generated samples demonstrate that all modes were sampled from although not uniformly. 
On the other hand, RDMC suffers from metastability on this target and struggles to sample from each mode. In such a case, post-processing is not able to improve the generation as the space does not get explored appropriately as with IS. Panels (c) and (d) repeat this analysis on the $40$-mode mixture, whose modes are closer together but more numerous. Plotting the first two coordinates at $d=30$ shows that both IS and RDMC place almost no samples on the modes, while their post-processed versions do 
and stay on par with the classical chains.

\begin{figure}
\centering
\includegraphics[width=.95\textwidth]{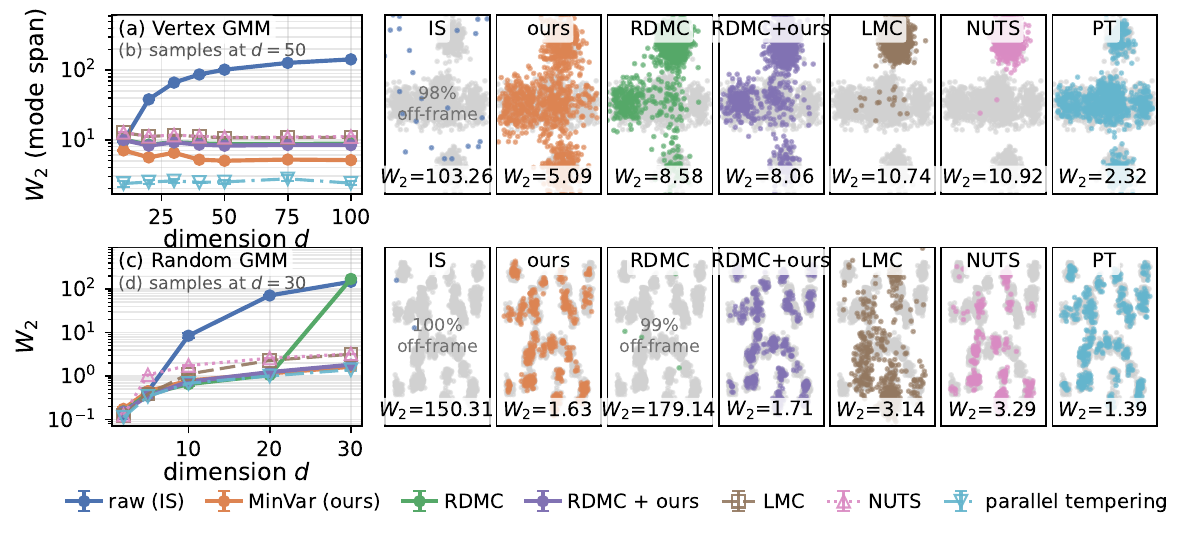}
\vspace{-\baselineskip}
\caption{\textit{$W_2$ against dimension and  samples.}
(a,b) the hypercube mixture and (c,d) the $40$-mode GMM. In high dimensions base
diffusion samplers miss most modes; post-processing finds them.}
\label{fig:vertex-dim}
\end{figure}

\paragraph{Where the method stops improving upon the base estimator.} 

We evaluate the sampler's ability to sample under increasing 
mode \emph{separation}, using the two-component anisotropic mixture with $0.7/0.3$ weights and
sweeping the Mahalanobis distance between its modes. At lower levels 
of separation our method significantly outperforms the base sampler. We also see the post processing vector is spreading mass across signal levels at low separation levels in Figure~\ref{fig:sep-separation} (b) and (c). However, as observed in Figure~\ref{fig:sep-separation} panels (d) and (e), as the separation increases, the Wiener oracle estimation collapses to the on-policy estimator. In this case, spending the full estimation budget on-policy offers slight improvement over splitting the budget and post-processing: Figure~\ref{fig:sep-separation} (a) shows on-policy computation matched estimation with lower $W_{2}$ levels.

\vspace{-\baselineskip}
\begin{figure}
\centering
\includegraphics[width=\textwidth]{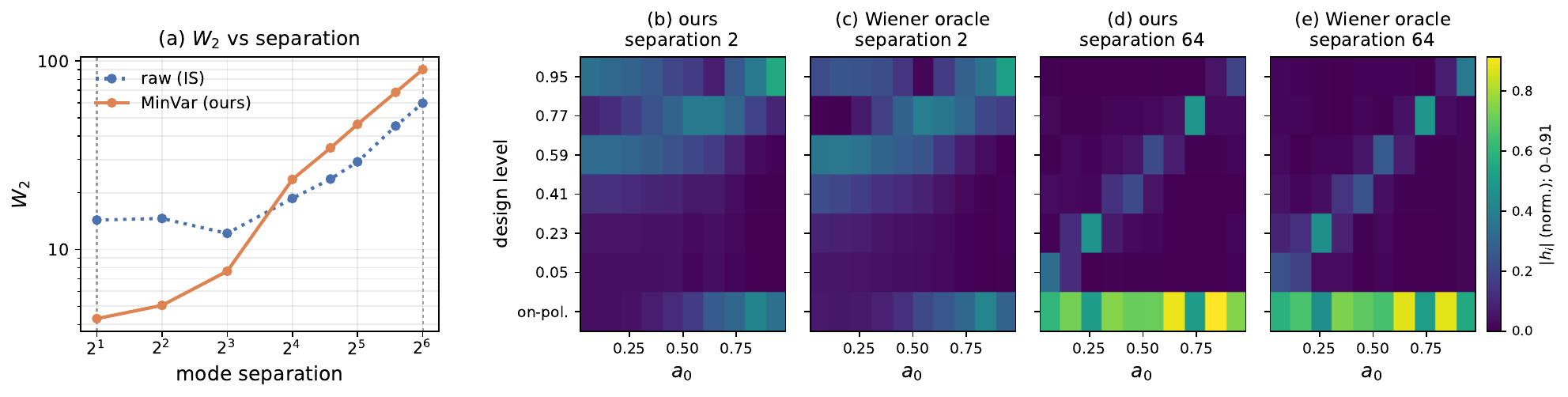}
\vspace{-\baselineskip}
\caption{Two-component mixture at $d=10$. (a) $W_2$ against mode separation. (b)--(e) Weights
$|h_i|$ per design level (rows) and query level $a_0$ (columns). At separation $64$ the Wiener
oracle puts almost all its weight on the on-policy query, so no affine post-processor can
improve on it.}
\label{fig:sep-separation}
\end{figure}

\section{Conclusion and Future Works}

We introduce an affine post-processing approach to score estimation.
We demonstrate, with a random score estimator with nonzero bias
and variance, on-policy estimation is almost never the best use of a fixed computation
budget for score estimation. Splitting a compute budget across queries and fitting an affine combination of the values gives a lower error. This allows us to scale diffusion-based approaches to higher dimensions, while retaining the benefits of the score-based samplers. 

Post-processing is an additional step built on top of a base score estimator. When the base queries become so corrupted they supply little useful off-policy information, the optimal post-processing rule can collapse to pure on-policy estimation. In this case, the sample budget spent on off-policy queries is not utilized, and some small improvement is seen by spending the full estimation budget on-policy to slightly increase the effective sample size. For this reason, improving the base estimator's exploration of the space is an area that still admits
further exploration. Future work includes choosing the design
levels adaptively, and deriving theoretical bounds for score estimation error for the practical, non-oracle rule.

\newpage

\subsection*{AI use statement}

Generative AI was utilized for verifying proofs, formulating mathematical claims, assisting in the writing of proofs, suggesting edits to mathematical formulations, and designing or providing feedback on research methodology and experiments. All theorems and proofs originated as constructs by the authors, with AI-suggested edits and changes added after the fact. Proofs are verified as correct by the authors and also independently verified by AI checks. Proofs are never copied entirely from AI sources without modification. Generative AI proofs are read by the authors as secondary sources, and relevant suggestions incorporated into a preexisting human generated proof as an improvement for clarity.

The decomposition of the population-wide error into a quadratic form in Theorem 1 was first constructed by the authors, but AI assisted literature review suggested the name of the Wiener form and connection to past literature.

Additionally, Generative AI was used for summarizing or analyzing existing literature, discovering research topics or identifying gaps, brainstorming, sourcing/searching for information, editing 
to improve readability, and identifying relevant literature.

Generative AI was not utilized to generate synthetic data sets, help develop theoretical models or conceptual frameworks, provide critical ingredients for proving mathematical claims, propose or refine hypotheses, assist with translation, clean and reformat datasets, support qualitative and thematic data analysis, or interpret results. The distinction between ``provide critical ingredients for proving mathematical claims'' and ``assist in the writing of proofs'' for the authors is that no mathematical formulations originated as a purely Generative AI construct, but all contributions added on top of an original human generated concept or proof as a refinement.

For the experimental code, the authors first wrote a base implementation of the estimators, samplers, and evaluation pipeline. Generative AI was subsequently used to write the scripts that launch and organize experimental runs (e.g., sweeps over hyperparameters, seeds, and problem dimensions), to revise the existing code base, and to produce the plotting code. All AI-generated or AI-modified code was reviewed by the authors. Generative AI was also used as a writing aid, offering edits to improve the clarity and readability of the text. 

We take responsibility for the final content of this work, including text, claims or artifacts produced with the aid of generative AI.

\subsubsection*{Acknowledgments}
We gratefully acknowledge the support of the NSF for FODSI through grant DMS-2023505, of the NSF and the Simons Foundation for the Collaboration on the Theoretical Foundations of Deep Learning through awards DMS-2031883 and 814639, and of the ONR through MURI award N000142112431, and a research grant from Bridgewater to the Simons Institute.

\bibliography{affine_pp_refs}
\bibliographystyle{plainnat}

\appendix
\section{Appendix}
\subsection{Oracle Smoother Theory Proofs and Supporting Lemmas}

\subsubsection*{Proof of Theorem \ref{wiener_oracle}}
\begin{proof}
Fix the on-policy signal level $a_{0}$, the grid of off-policy signal levels $(a_{i})_{i=1}^{N+2}$, and a coordinate $j$ whose index is suppressed in the notation throughout the proof. Note that $z$ is not a $U$ dependent object. Denote random variables for the score estimates as $\hat{s}_{i}(u) = \hat{s}(u,a_{i})$, and let $s(u) = s(u,a_{0})$ be the true score.
The expected $L^{2}$ error for the affine post-processing estimator is then
\begin{align}
    J(z, a_{0}) = E_{U \sim p_{a_{0}}}[E[(s(U) - z^{T}\begin{pmatrix} \hat{s}(U)\\ -1 \end{pmatrix})^{2}]],
\end{align}
where the outer expectation is over locations $U$ distributed according to the marginal $p_{a_{0}}$, and the inner expectation is over the internal randomness of the score estimator at a fixed $U$ value. Assume at a fixed $U$ all $\hat{s}_{i}(U)$ are independent queries to the score estimator and are independent random variables. Perform a bias-variance decomposition on the inner expectation over internal randomness,
\begin{align}
    E_{U \sim p_{a_{0}}}[(s(U) - z^{T}\begin{pmatrix} \tilde{m}(U)\\ -1 \end{pmatrix})^{2}+\text{Var}(\tilde{h}^{T}\hat{s})].
\end{align}
By the independence of the $\hat{s}$ estimators the variance of the sum decomposes into the sum of the variances. Expand the squared bias term and we have expression,
\begin{align}
    E_{U \sim p_{a_{0}}}[z^{T}\begin{pmatrix} \tilde{m}(U)\\ -1 \end{pmatrix}\begin{pmatrix} \tilde{m}(U)\\ -1 \end{pmatrix}^{T}z-2 z^{T}\begin{pmatrix} s(U)\tilde{m}(U)\\ -s(U)\end{pmatrix}+s(U)^{2}+h^{T}V(u,a_{i})h].
\end{align}
Bring the non-$U$ dependent object $z$ outside the $U$ expectation and this is expressed as a quadratic in $z$,
\begin{align}
&z^{T}\begin{pmatrix}
E_{U\sim p_{a_{0}}}[\tilde m \tilde m^{\top}] + \mathcal{V} & -E_{U\sim p_{a_{0}}}[\tilde m] \\
-E_{U\sim p_{a_{0}}}[\tilde m]^{\top} & 1
\end{pmatrix}z+z^{T} \begin{pmatrix}
    -2 E_{U \sim p_{a_{0}}}[s \tilde{m}]\\
    2 E_{U \sim p_{a_{0}}}[s]
\end{pmatrix}+E_{U \sim p_{a_{0}}}[s^{2}]\\
&= z^{T}\mathcal{M}(a_{0})z+z^{T}\mu(a_{0})+E_{U \sim p_{a_{0}}}[s^{2}].
\end{align}
Note the score is mean 0 with respect to its own density by fundamental theorem of calculus,
\begin{align}
    E_{U \sim p_{a_{0}}}[s]&= \int \partial_{u_{j}}p_{a_{0}}(u)du = 0
\end{align}
as such the first vector of terms in $\mu(a_{0})$ is a proper covariance, a product with a mean 0 random variable, and the final term is 0.

Assuming $\mathcal{M}(a_{0})$ is a full rank matrix, the expression is strictly convex function of $z$ and has a unique minimizer. Taking the $z$ gradient and setting it to zero gives condition for the critical point
\begin{align}
    \mathcal{M}(a_{0})z+\frac{1}{2}\mu(a_{0})&=0,
\end{align}
which gives conditions on $\tilde{h}, \nu$ of
\begin{align}
    (E_{U \sim p_{a_{0}}}[\tilde{m}\tilde{m}^{T}]+\mathcal{V})\tilde{h}-E_{U \sim p_{a_{0}}}[\tilde{m}]\nu-\text{Cov}_{U \sim p_{a_{0}}}[\tilde{m},s]&=0 \label{critical_1}\\
    - \tilde{h}^{T}E_{U \sim p_{a_{0}}}[\tilde{m}]+\nu &=0.
\end{align}
Solving for $\nu$ in terms of $\tilde{h}$ we must have
\begin{align}
    \nu&= E_{U \sim p_{a_{0}}}[\tilde{m}]^{T}\tilde{h}.
\end{align}
Substituting this expression for $\nu$ back into equation \eqref{critical_1} turns the second moment term $E_{U \sim p_{a_{0}}}[\tilde{m}\tilde{m}^{T}]$ into  a proper covariance term $\text{Cov}_{U \sim p_{a_{0}}}[\tilde{m}]$, which yields solution for $\tilde{h}$ as
\begin{align}
    \tilde{h}^{*}&= \big[ \mathrm{Cov}_{U\sim p_{a_{0}}}(\tilde m) + \mathcal{V} \big]^{-1} \mathrm{Cov}_{U\sim p_{a_{0}}}\big(\tilde m, s\big).
\end{align}
\end{proof}

\begin{lemma}[Gradient of $J$]
For any $\tilde{h} \in \mathbb{R}^{N+3}$ and $\nu \in \mathbb{R}$, the partial derivatives of $J$ are,
\begin{align}
\partial_{\tilde{h}_{i}}J(\tilde{h},\nu,a_{0}) &= 2\Big(\sum_{j=0}^{N+2}E_{U \sim p_{a_{0}}}[\tilde{m}_{i}\tilde{m}_{j}]\,\tilde{h}_{j} + \mathcal{V}_{i,i}\tilde{h}_{i} - E_{U \sim p_{a_{0}}}[\tilde{m}_{i}]\,\nu - \text{Cov}_{U \sim p_{a_{0}}}[\tilde{m}_{i}, s]\Big), \\
\partial_{\nu}J(\tilde{h},\nu,a_{0}) &= -2\Big(\sum_{j=0}^{N+2}E_{U \sim p_{a_{0}}}[\tilde{m}_{j}]\,\tilde{h}_{j} - \nu \Big).
\end{align}
Since $J$ is a convex quadratic, $(h,\nu)$ is the Wiener estimator if and only if all $N+4$ derivatives vanish.
\end{lemma}
\begin{proof}
    This expression follows from the gradient of a quadratic term
    \begin{align}
        \nabla_{z} J(z,a_{0}) = 2 \mathcal{M}(a_{0})z+\mu(a_{0}).
    \end{align}
    Reading the expressions coordinate by coordinate and using the forms of $\mathcal{M}(a_{0}), \mu(a_{0})$ as in Theorem \ref{wiener_oracle} gives the equivalent form.
\end{proof}

\subsubsection*{Proof of Corollary \ref{off_policy_optimal}}
\begin{proof}
This corollary follows from the gradient expression lemma specialized to $z_{\text{on}} = (e_{0}, 0)$. Recall the on-policy score is mean 0 under the density. The partial derivatives specialize to
\begin{align}
&\partial_{\tilde{h}_{i}}J(z_{\text{on}},a_{0}) = 2\Big(E_{U \sim p_{a_{0}}}[\tilde{m}_{i}\tilde{m}_{0}] + \mathcal{V}_{0,0}1\{i = 0\} - \text{Cov}_{U \sim p_{a_{0}}}[\tilde{m}_{i}, s]\Big)\\
&= 2 (E_{U \sim p_{a_{0}}}[\tilde{m}(U,a_{i})(B(u,a_{0})+s(u,a_{0}))]-E_{U \sim p_{a_{0}}}[\tilde{m}(U,a_{i})(s(u,a_{0}))])+2\mathcal{V}_{0,0}1\{i= 0\}\\
&= 2(E_{U \sim p_{a_{0}}}[\tilde{m}_{i}B]+\mathcal{V}_{0,0}1\{i = 0\}).
\end{align}
For $i \neq 0$ this is equal to $2E_{U \sim p_{a_{0}}}[B \tilde{m}_{i}]$. For $i = 0$, the on-policy expected $L^{2}$ error is exactly
\begin{align}
    J(z_{\text{on}},a_{0}) = E_{U \sim p_{a_{0}}}[B^{2}]+\mathcal{V}_{0,0},
\end{align}
and the derivative can be expressed as $2 (E_{U \sim p_{a_{0}}}[B s]+J(z_{\text{on}}, a_{0}))$. The $\nu$ derivative is expressed as,
\begin{align}
\partial_{\nu}J(\tilde{h},\nu,a_{0}) &= -2 E_{U \sim p_{a_{0}}}[\tilde{m}_{0}] = -2 E_{U \sim p_{a_{0}}}[B].
\end{align}
\end{proof}
\subsubsection*{Proof of Corollary \ref{unbiased_wiener}}
\begin{proof}
We take the expressions for the derivative from Corollary \ref{off_policy_optimal} and set $B = 0$ in all cases. Setting the gradient equal to 0, the conditions for the on-policy estimator to be optimal are already satisfied for the $\tilde{h}_{i}, i \neq 0$ derivatives and the $\nu$ derivative. The only derivative condition to check is the $\tilde{h}_{0}$ derivative. This derivative is equal to 0 if and only if
\begin{align}
    J(z_{\text{on}}, a_{0})&=0.
\end{align}
Since there is no bias by assumption, this $L^{2}$ error is only on-policy variance $\mathcal{V}_{0,0}$. This population variance is non-negative, and thus the on-policy estimator is the optimal solution if and only if
\begin{align}
    \mathcal{V}_{0,0} = E_{U \sim p_{a_{0}}}[\text{Var}(\hat{s}(U,a_{0}))] = 0,
\end{align}
which means the estimator has no variance. Combined with the unbiased assumption, this would mean the estimator $\hat{s}$ returns perfect score readings. For any random, unbiased estimator with some variance so that $\mathcal{V}_{0,0}>0$, the on-policy estimator is not Wiener optimal.
\end{proof}

\subsection{Additional Ridge Regression Theory}\label{lin_reg_theory}

To further connect the two forms of $\beta^{*}$ as a coefficient in the functional basis and $h^{*},g^{*}$ as linear smoother directly on the score estimates, we consider the penalty matrices $\Sigma$ and $\Omega$ as diagonal. We will construct the equivalent ridge regression and adjoint optimization problem in Theorem \ref{ridge_adjoint} by first starting with the unconstrained ridge regression solution and then imposing the terminal constraints.

We first review standard ridge regression constructions. Assume initially that $g = 0$ so we do not impose any terminal constraints and we only fit the $h$ vector. Consider that for a fixed $u$ and fixed dimension $j$, $s_{j}(u,a)$ is a continuous function in $a \in [0,1]$. Each $\hat{s}_{j}(u,a_{i})$ is a noisy observation of that function, and we can perform ridge regression in a finite $K$ basis, which amounts to minimizing the quadratic objective
\begin{align}
    \min_{\beta} (\hat{s} - \Phi \beta)^{T}\Sigma^{-1}(\hat{s} - \Phi \beta)+\lambda \beta^{T}\Omega^{-1} \beta.
\end{align}
The minimizing parameter via standard ridge regression theory is the vector of estimates put through the linear operator,
\begin{align}
    \beta^{*} = (\Phi^{T}\Sigma^{-1}\Phi+ \lambda\Omega^{-1})^{-1}\Phi^{T}\Sigma^{-1}\hat{s}.
\end{align}
The vector of fit values at each of the observed $a_{i}$ locations is then the $\hat{s}$ vector put through the shrinkage projection matrix $H$,
\begin{align}
    \bar{s} = \Phi \beta^{*}&= \Phi(\Phi^{T}\Sigma^{-1}\Phi+ \lambda\Omega^{-1})^{-1}\Phi^{T}\Sigma^{-1}\hat{s} \equiv H \hat{s}.
\end{align}
Define $b$ as the first row of the $\Phi$ matrix. The evaluation of $\bar{s}$ at the $a_{0}$ signal level, which is the first entry of the $\bar{s}$ vector, is the inner product with the fixed vector $h^{*}$, which is the first row of the matrix $H$,
\begin{align}
&\bar{s}(u,a_{0}) = \langle h^{*}, \hat{s}(u) \rangle&h^{*} = \Sigma^{-1}\Phi(\Phi^{T}\Sigma^{-1}\Phi+ \lambda\Omega^{-1})^{-1}b.
\end{align}

We can also construct this $h^{*}$ as the solution to its own minimization problem. It is first helpful to decompose the pointwise bias of the affine post-processing estimator into two terms, the transmitted bias and the reproduction error.

\begin{lemma}[Bias--variance decomposition]\label{bv_lemma}
For an affine post-processor with weights $h$, define the transmitted bias and reproduction error
\begin{align}
T_{B}(u,a_{0}) &= \sum_{i=0}^{N}h_{i}B(u,a_{i})&\rho(u,a_{0}) = \sum_{i=0}^{N}h_{i}s(u,a_{i}) - s(u,a_{0}).
\end{align}
Then, by independence of the queries,
\begin{align}
E[\bar{s}(u,a_{0})]-s(u,a_{0}) = T_{B}(u,a_{0})+\rho(u,a_{0}), \qquad
\mathrm{Var}(\bar{s}(u,a_{0})) = \sum_{i=0}^{N}h_{i}^{2}V(u,a_{i}),
\end{align}
so the expected squared error at $(u,a_{0})$ is $(T_{B}+\rho)^{2}+\sum_{i}h_{i}^{2}V(u,a_{i})$. Post-processing improves on the on-policy estimate at $(u,a_{0})$ exactly when this is below $B(u,a_{0})^{2}+V(u,a_{0})$, the value at $h=e_{0}$.
\end{lemma}

Assume that the score estimator is unbiased, then there is no transmitted bias $T_{B}$ and the bias is fully the reproduction error. Set $\lambda = 0$. Then if $\Sigma = I$, $b = \Phi^{T}e_{0}$, $h$ is the projection of $e_{0}$ into the span of the design matrix,
\begin{align}
    h = \Phi(\Phi^{T}\Phi)^{-1}\Phi^{T}e_{0}.
\end{align}

Assume the true score is expressed in the functional basis as
\begin{align}
    s(u,a) = \sum_{k \geq 0}\alpha_{k}(u)\phi_{k}(a).
\end{align}
Then assume it has no terms $k>K$. This assumes the true score vector at the on-policy and off-policy locations is in the span of the design matrix. Define the defect terms
\begin{align}
    \delta_{k}(a) =  \sum_{i=0}^{N}h_{i}\phi_{k}(a_{i}) - \phi_{k}(a_{0}) = (\Phi^{T} h)_{k} - b_{k}.
\end{align}
This term penalizes the $h$ smoother for not perfectly recreating proper elements of the finite $K$ basis. If the true score was exactly represented in the finite $K$ basis $s = \sum_{k=0}^{K}\alpha_{k}\phi_{k}$, the reproduction error which determines the bias is $\rho(u,a) = \sum_{k=0}^{K}\delta_{k}(a)\alpha_{k}(u)$. Thus the minimization problem,
\begin{align}
    \min_{h} \|b - \Phi^{T}h\|^{2}
\end{align}
is exactly penalizing $h$ for introducing defect terms. This expression penalizes each defect degree $k$ equally, so a diagonal matrix $\Omega$ penalizes different degree $k$ defects differently,
\begin{align}
    \min_{h} \sum_{k=0}^{K}\Omega_{k,k}((\Phi^{T}h)_{k} - b_{k})^{2}.
\end{align}
We then introduce a $\lambda$ penalty on the norm of $h$ in a $\Sigma$ norm,
\begin{align}
    \min_{h} \sum_{k=0}^{K}\Omega_{k,k}((\Phi^{T}h)_{k} - b_{k})^{2}+\lambda \sum_{i=0}^{N}\Sigma_{i,i}h_{i}^{2}.
\end{align}
$\lambda = 0$ forces the $h$ vector to only focus on minimizing the defect terms, which amounts in a hard constraint $(\Phi^{T} h)_{k} = b_{k}$ for each $k$. By introducing a non-zero $\lambda$, the $h$ vector shrinks its norm while introducing defects into the reproduction term. This is a pure bias-variance tradeoff. The solution here is a standard quadratic program, though not as directly stated as the ridge regression theory for $\beta^{*}$, so we state the result as a proper lemma.

\begin{lemma}\label{h_solution}
Let $\Phi \in \mathbb{R}^{(N+1)\times(K+1)}$, let $\Sigma \in \mathbb{R}^{(N+1)\times(N+1)}$ and $\Omega \in \mathbb{R}^{(K+1)\times(K+1)}$ be symmetric positive definite, and let $\lambda > 0$. The problem
\begin{align}
\min_{h \in \mathbb{R}^{N+1}}\ \big(\Phi^{T}h - b\big)^{T}\Omega\big(\Phi^{T}h - b\big) + \lambda\,h^{T}\Sigma h
\end{align}
has the unique solution
\begin{align}
h^{*} = \big(\Phi\,\Omega\,\Phi^{T} + \lambda\Sigma\big)^{-1}\Phi\,\Omega\, b
= \Sigma^{-1}\Phi\big(\Phi^{T}\Sigma^{-1}\Phi + \lambda\Omega^{-1}\big)^{-1} b.
\end{align}
\end{lemma}

\begin{proof}
Write $C = \Phi\Omega\Phi^{T} + \lambda\Sigma$ and $A = \Phi^{T}\Sigma^{-1}\Phi + \lambda\Omega^{-1}$. Since $\Phi\Omega\Phi^{T}$ and $\Phi^{T}\Sigma^{-1}\Phi$ are positive semi-definite and $\lambda\Sigma$, $\lambda\Omega^{-1}$ are positive definite, both $C$ and $A$ are positive definite, hence invertible.

The objective is a quadratic in $h$ with Hessian $2C \succ 0$, so it is strictly convex and its unique minimizer solves the first-order condition
\begin{align}
2\,\Phi\,\Omega\big(\Phi^{T}h - b\big) + 2\lambda\,\Sigma h = 0
\quad\Longleftrightarrow\quad
C h = \Phi\,\Omega\, b,
\end{align}
so $h^{*} = C^{-1}\Phi\,\Omega\, b$.

For the second form, expand
\begin{align}
\Phi\,\Omega\, A
= \Phi\,\Omega\,\Phi^{T}\Sigma^{-1}\Phi + \lambda\,\Phi
= \big(\Phi\,\Omega\,\Phi^{T} + \lambda\Sigma\big)\Sigma^{-1}\Phi
= C\,\Sigma^{-1}\Phi.
\end{align}
Multiplying on the left by $C^{-1}$ and on the right by $A^{-1}$ gives $C^{-1}\Phi\,\Omega = \Sigma^{-1}\Phi A^{-1}$. Applying both sides to $b$ gives $h^{*} = \Sigma^{-1}\Phi A^{-1} b$ as an equivalent form, which is also the expression for the $h$ smoother derived from the ridge regression problem with $\beta^{*}$.
\end{proof}

Penalizing each defect term equally $\Omega_{k,k} = 1$ corresponds to penalizing the norm of each $\beta_{k}$ in ridge regression equally. An identity $\Sigma_{i,i} = 1$ matrix corresponds to unweighted regression. Thus, we have paired minimization problems,
\begin{align}
    &\min_{h} \sum_{k=0}^{K}\Omega_{k,k}( (\Phi^{T}h)_{k} - b_{k})^{2}+\lambda\sum_{i=0}^{N}h_{i}^{2}\Sigma_{i,i}\\
    &\min_{\beta} \sum_{i=0}^{N} \frac{1}{\Sigma_{i,i}}(\hat{s}(u,a_{i}) - \sum_{k=0}^{K} \beta_{k} \phi_{k}(a_{i}))^{2}+\lambda\sum_{k=0}^{K}\frac{1}{\Omega_{k,k}}\beta_{k}^{2}
\end{align}

We then specialize to terminal score constraints by setting certain $\Sigma_{i,i}$ values to 0, and hard defect constraints by setting setting certain $\Omega_{k,k}$ penalties to infinity.

Consider extending the $\hat{s}$ vector of score observations with two more values $s(u,0)$ and $s(u,1)$ as perfect observations of the score terminal values. Include two more coefficients $\tilde{h} = (h,g)$ in the smoothing vector. Then these $\hat{s}_{N+1} = s(u,0), \hat{s}_{N+2} = s(u,1)$ act no differently in the above equations than any other observation. Define $\Psi_{0}$ and $\Psi_{1}$ as the vector of the basis functions evaluated at $0$ and $1$. We can write equivalent optimization problems with additional terms singling out the contribution of the $g$ term and these additional two observations
\begin{align}
    &\min_{h,g} \sum_{k=0}^{K}\Omega_{k,k}( (\Phi^{T}h)_{k}+(\Psi^{T}g)_{k} - b_{k})^{2}+\lambda\sum_{i=0}^{N}h_{i}^{2}\Sigma_{i,i}+\lambda \sum_{i=0}^{1}\Sigma_{N+i, N+i} g_{i}^{2}\\
    &\min_{\beta} \sum_{i=0}^{N} \frac{1}{\Sigma_{i,i}}(\hat{s}(u,a_{i}) - \sum_{k=0}^{K} \beta_{k} \phi_{k}(a_{i}))^{2}+\lambda\sum_{k=0}^{K}\frac{1}{\Omega_{k,k}}\beta_{k}^{2}+ \sum_{i=0}^{1}\frac{1}{\Sigma_{N+i, N+i}}(s(u,i)-\Psi_{i}^{T}\beta)^{2}.
\end{align}

Setting $\Sigma_{N+1, N+1}, \Sigma_{N+2, N+2}$ to 0 forces a hard constraint on the $\beta$ basis to exactly match the terminal constraints, and results in $g$ as a slack variable with no $\Sigma$ penalty in the $h,g$ minimization problem.

\begin{lemma}\label{terminal_constraints}
Let $\Psi \in \mathbb{R}^{2\times(K+1)}$ contain the basis evaluated at the terminal levels, $\Psi_{0,k} = \phi_{k}(0)$, $\Psi_{1,k} = \phi_{k}(1)$, assume $\Psi$ has full row rank, and let $v(u) = (s(u,0), s(u,1))$. Append the terminal scores to the observations as extra rows of the design matrix, with smoothing vector $\tilde h = (h,g)$ and variances $\Sigma_{N+1,N+1} = \Sigma_{N+2,N+2} = \varepsilon > 0$. Applying Lemma~\ref{h_solution} to this extended problem gives the paired problems
\begin{align}
&\min_{h,g}\ \big(\Phi^{T}h + \Psi^{T}g - b\big)^{T}\Omega\big(\Phi^{T}h + \Psi^{T}g - b\big) + \lambda\, h^{T}\Sigma h + \lambda\varepsilon\,\|g\|^{2},\\
&\min_{\beta}\ \big(\hat{\mathbf{s}}(u) - \Phi\beta\big)^{T}\Sigma^{-1}\big(\hat{\mathbf{s}}(u) - \Phi\beta\big) + \frac{1}{\varepsilon}\big\|v(u) - \Psi\beta\big\|^{2} + \lambda\,\beta^{T}\Omega^{-1}\beta.
\end{align}
As $\varepsilon \to 0$, their solutions converge to the unique solutions of
\begin{align}
\beta^{*}(u) &= \operatorname*{arg\,min}_{\beta \in \mathbb{R}^{K+1}}\ \big(\hat{\mathbf{s}}(u) - \Phi\beta\big)^{T}\Sigma^{-1}\big(\hat{\mathbf{s}}(u) - \Phi\beta\big) + \lambda\,\beta^{T}\Omega^{-1}\beta \quad \text{s.t.}\quad \Psi\beta = v(u),\\
(h^{*},g^{*}) &= \operatorname*{arg\,min}_{h \in \mathbb{R}^{N+1},\, g \in \mathbb{R}^{2}}\ \big(\Phi^{T}h + \Psi^{T}g - b\big)^{T}\Omega\big(\Phi^{T}h + \Psi^{T}g - b\big) + \lambda\, h^{T}\Sigma h,
\end{align}
so the terminal values are reproduced exactly and $g$ is unpenalized. Thus, linear constraints in the functional ridge regression yield free parameters in the adjoint problem.
\end{lemma}

Furthermore, dropped $\beta$ penalties via $\Omega_{k,k} \to \infty$ result in hard constraints on the defect terms. In our problem, we impose hard constraints on the 0 order defect terms. We penalize degree $k\geq 1$ $\beta_{k}$ terms with a cubic penalty $\Omega_{k,k} = \frac{1}{k^{3}}$.

\begin{lemma}\label{dropped_penalty}
In the setting of Lemma~\ref{terminal_constraints}, assume in addition that $\Phi$ has full column rank ($N \geq K$) and that $\Omega$ is diagonal, and let $\mathcal{K} \subset \{0,\dots,K\}$ index coefficients that are left unpenalized: $\Omega_{k,k} \in (0,\infty)$ for $k \notin \mathcal{K}$ and $\Omega_{k,k} = \infty$ for $k \in \mathcal{K}$, with the convention $1/\infty = 0$, so that $\Omega^{-1}$ is diagonal with $(\Omega^{-1})_{k,k} = 0$ for $k \in \mathcal{K}$. Consider
\begin{align}
\beta^{*}(u) = \operatorname*{arg\,min}_{\beta \in \mathbb{R}^{K+1}}\ \big(\hat{\mathbf{s}}(u) - \Phi\beta\big)^{T}\Sigma^{-1}\big(\hat{\mathbf{s}}(u) - \Phi\beta\big) + \lambda\sum_{k \notin \mathcal{K}}\frac{\beta_{k}^{2}}{\Omega_{k,k}} \quad \text{s.t.}\quad \Psi\beta = v(u),
\end{align}
and, writing $\delta(h,g) = \Phi^{T}h + \Psi^{T}g - b$ for the defect vector,
\begin{align}
(h^{*},g^{*}) = \operatorname*{arg\,min}_{h \in \mathbb{R}^{N+1},\, g \in \mathbb{R}^{2}}\ \sum_{k \notin \mathcal{K}}\Omega_{k,k}\,\delta_{k}(h,g)^{2} + \lambda\, h^{T}\Sigma h \quad \text{s.t.}\quad \delta_{k}(h,g) = 0,\ k \in \mathcal{K}.
\end{align}
Unpenalized coefficients in the ridge regression become hard constraints on the defects of the linear smoother. Both problems have unique solutions, and for every realization of $\hat{\mathbf{s}}(u)$,
\begin{align}
\bar{s}(u,a_{0}) = \sum_{k=0}^{K}\beta^{*}_{k}(u)\,\phi_{k}(a_{0}) = \sum_{i=0}^{N}h^{*}_{i}\,\hat{s}(u,a_{i}) + g^{*}_{0}\,s(u,0) + g^{*}_{1}\,s(u,1).
\end{align}
Writing $A = \Phi^{\top}\Sigma^{-1}\Phi + \lambda\,\Omega^{-1}$, $G = \Psi A^{-1}\Psi^{\top}$ and $P = I - A^{-1}\Psi^{\top}G^{-1}\Psi$, the solutions are
\begin{align}
\beta^{*}(u) &= P A^{-1}\Phi^{\top}\Sigma^{-1}\,\hat{\mathbf{s}}(u) + A^{-1}\Psi^{\top}G^{-1}\,v(u),\\
h^{*} &= \Sigma^{-1}\Phi A^{-1}P^{\top} b, \qquad
g^{*} = G^{-1}\Psi A^{-1} b.
\end{align}
This proves the form of the problem stated in Theorem~\ref{ridge_adjoint}.
\end{lemma}
\begin{proof}
Throughout, $\Sigma \succ 0$, $\lambda > 0$, $\Phi$ has full column rank and $\Psi$ has full row rank. Write $c(u) = \Phi^{T}\Sigma^{-1}\hat{\mathbf{s}}(u)$. Both forms of $\beta^{*}$ and $(h^{*}, g^{*})$ follow from the KKT conditions applied to linear constrained convex minimization problems.

\noindent\textbf{Invertibility of Matrices}\\
Since $\Phi$ has full column rank, the matrix $\Phi^{T}\Sigma^{-1}\Phi \succ 0$ is positive definite, while the matrix $\lambda\Omega^{-1} \succeq 0$ is positive semi-definite with some 0 entries on the diagonal. Therefore the matrix $A = \Phi^{T}\Sigma^{-1}\Phi + \lambda\Omega^{-1} \succ 0$ is positive definite and invertible. Since $\Psi$ has full row rank and $A^{-1} \succ 0$, also $G = \Psi A^{-1}\Psi^{T} \succ 0$ is positive definite and invertible. Therefore, $A$, $G$, $A^{-1}$, $G^{-1}$, and $P$ are well defined.

\noindent\textbf{The $\beta$ optimization problem}.\\ Expanding the $\beta$ form, the objective is $\beta^{T}A\beta - 2\beta^{T}c(u)$ plus a constant independent of $\beta$. It is strictly convex since $A \succ 0$, and the constraint $\Psi\beta = v(u)$ is affine and feasible since $\Psi$ has full row rank. So the minimizer exists, is unique, and is characterized by the KKT conditions. The KKT conditions for this problem state for some vector $\nu \in \mathbb{R}^{2}$, the pair $(\beta, \nu)$ solves the equations:
\begin{align}
A\beta - c(u) = \Psi^{T}\nu, \qquad \Psi\beta = v(u).
\end{align}
The first equation gives $\beta = A^{-1}\big(c(u) + \Psi^{T}\nu\big)$. Substituting into the second gives $\Psi A^{-1}c(u) + G\nu = v(u)$, so $\nu = G^{-1}\big(v(u) - \Psi A^{-1}c(u)\big)$. Therefore
\begin{align}
\beta^{*}(u) = A^{-1}c(u) + A^{-1}\Psi^{T}G^{-1}\big(v(u) - \Psi A^{-1}c(u)\big)
= P A^{-1}c(u) + A^{-1}\Psi^{T}G^{-1}v(u),
\end{align}
using $A^{-1} - A^{-1}\Psi^{T}G^{-1}\Psi A^{-1} = \big(I - A^{-1}\Psi^{T}G^{-1}\Psi\big)A^{-1} = PA^{-1}$.

\noindent \textbf{The adjoint problem}\\
 Let $D$ be the diagonal matrix with $D_{k,k} = \Omega_{k,k}$ for $k \notin \mathcal{K}$ and $D_{k,k} = 0$ for $k \in \mathcal{K}$, so the objective is $F(h,g) = \delta^{T}D\delta + \lambda h^{T}\Sigma h$ with $\delta = \delta(h,g)$. $F$ is convex (as $D \succeq 0$ and $\Sigma \succ 0$) and the constraints $\delta_{k} = 0$, $k \in \mathcal{K}$, are affine in $(h,g)$. For such problems the KKT conditions are both necessary and sufficient for optimality. Introducing multipliers $2\mu_{k}$ for the constraints and using $\partial\delta/\partial h = \Phi^{T}$, $\partial\delta/\partial g = \Psi^{T}$, the KKT conditions read
\begin{align}
\Phi\, w + \lambda\Sigma h = 0, \qquad \Psi\, w = 0, \qquad \delta_{k} = 0 \ \ (k \in \mathcal{K}),
\end{align}
where $w \in \mathbb{R}^{K+1}$ has entries $w_{k} = \Omega_{k,k}\delta_{k}$ for $k \notin \mathcal{K}$ and $w_{k} = \mu_{k}$ for $k \in \mathcal{K}$. Since $(\Omega^{-1})_{k,k} = 1/\Omega_{k,k}$ for $k \notin \mathcal{K}$ and $(\Omega^{-1})_{k,k} = 0$ for $k \in \mathcal{K}$, the definition of $w$ together with the constraints is equivalent to the single equation $\delta = \Omega^{-1}w$ with $w$ free on $\mathcal{K}$. Hence $(h,g)$ is optimal if and only if there exists $w \in \mathbb{R}^{K+1}$ with
\begin{align}
\text{(i)}\ \ \Phi\, w + \lambda\Sigma h = 0, \qquad
\text{(ii)}\ \ \Psi\, w = 0, \qquad
\text{(iii)}\ \ \Phi^{T}h + \Psi^{T}g - b = \Omega^{-1}w.
\end{align}
From (i), $h = -\lambda^{-1}\Sigma^{-1}\Phi\, w$. Substituting into (iii),
\begin{align}
-&\lambda^{-1}\Phi^{T}\Sigma^{-1}\Phi\, w + \Psi^{T}g - b = \Omega^{-1}w\\
\Longleftrightarrow&
\lambda^{-1}A\, w = \Psi^{T}g - b\\
\Longleftrightarrow&
w = -\lambda A^{-1}\big(b - \Psi^{T}g\big).
\end{align}
Substituting into (ii), $\Psi A^{-1}\big(b - \Psi^{T}g\big) = 0$, i.e.\ $G g = \Psi A^{-1}b$, so
\begin{align}
g = G^{-1}\Psi A^{-1}b = g^{*}.
\end{align}
Then $b - \Psi^{T}g^{*} = \big(I - \Psi^{T}G^{-1}\Psi A^{-1}\big)b = P^{T}b$, so $w = -\lambda A^{-1}P^{T}b$ and
\begin{align}
h = -\lambda^{-1}\Sigma^{-1}\Phi\, w = \Sigma^{-1}\Phi A^{-1}P^{T}b = h^{*}.
\end{align}
Every step was an equivalence with $A$ and $G$ invertible, so the KKT system has exactly one solution $(h^{*},g^{*},w)$. By sufficiency of the KKT conditions $(h^{*},g^{*})$ is a minimizer, and by necessity every minimizer solves the KKT system, so it is the unique minimizer. Condition (iii) with $(\Omega^{-1})_{k,k} = 0$ confirms $\delta_{k}(h^{*},g^{*}) = 0$ for $k \in \mathcal{K}$.
\end{proof}

\section{Measured bias and variance of the post-processed score}
\label{app:bias-variance}
The relevant term to control the terminal KL error of a score-based sampler is the expected $L^{2}$ error $J(\tilde{h}, \nu ,a)$ defined in \eqref{L2}. This squared error can be decomposed into a bias and variance. While it is clear the affine post-processing estimator reduces the overall squared error in most of our experiments, it is of interest if this reduction is mostly bias reduction, variance reduction, or a bias-variance trade off where the estimator slightly inflates bias to reduce variance and overall squared error. For any affine $(\tilde h, \nu)$ the error functional \eqref{L2} decomposes exactly as
\begin{align}
J(\tilde h,\nu,a_{0}) =
\underbrace{E_{U \sim p_{a_{0}}}\big[\big(\tilde h^{\top}\tilde m(U) - \nu - s(U)\big)^{2}\big]}_{\text{bias}^{2}}
\;+\;
\underbrace{\tilde h^{\top}\mathcal{V}\,\tilde h}_{\text{variance}},
\end{align}
because $\tilde m(U)$ is the conditional mean of the base estimates and $\mathcal{V}$ their
conditional covariance. The on-policy
rule $\tilde h = e_{0}$, $\nu = 0$ returns $\tilde{B}^{2}_{0} + \mathcal{V}_{00}$, the base
estimator's own bias and variance. Both terms are computed from the same estimated
quadratic forms that produce the $J$ curves of Figure~\ref{fig:geometry-j}. Here Figures~\ref{fig:bias-variance-xgmm} and~\ref{fig:bias-variance-rgmm} rearrange those
curves rather than reporting a separate experiment. At a grid of signal levels, the total $J$ error is scaled so the full bar is $100\%$ and that error is split into it's bias and variance proportion. We compare this to the bias and variance of the affine post processing vector scaled by the same value.

\begin{figure}[H]
\centering
\includegraphics[width=\textwidth]{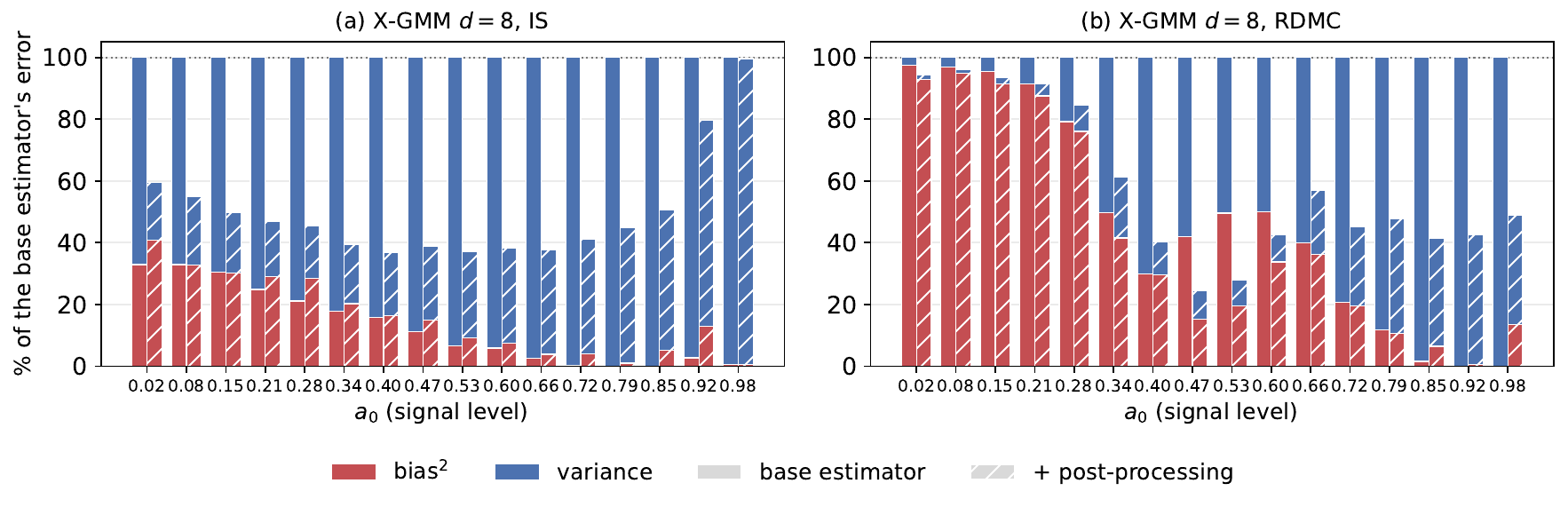}
\caption{Orthogonal-covariance mixture at $d=8$. Solid bars are the base estimator and
hatched bars the post-processed rule, each scaled so the base totals $100\%$ at that
signal level. Importance sampling (a) is variance-dominated at every level and
post-processing removes most of that variance; RDMC (b) is bias-dominated at low signal,
where little can be removed.}
\label{fig:bias-variance-xgmm}
\end{figure}

\begin{figure}[H]
\centering
\includegraphics[width=\textwidth]{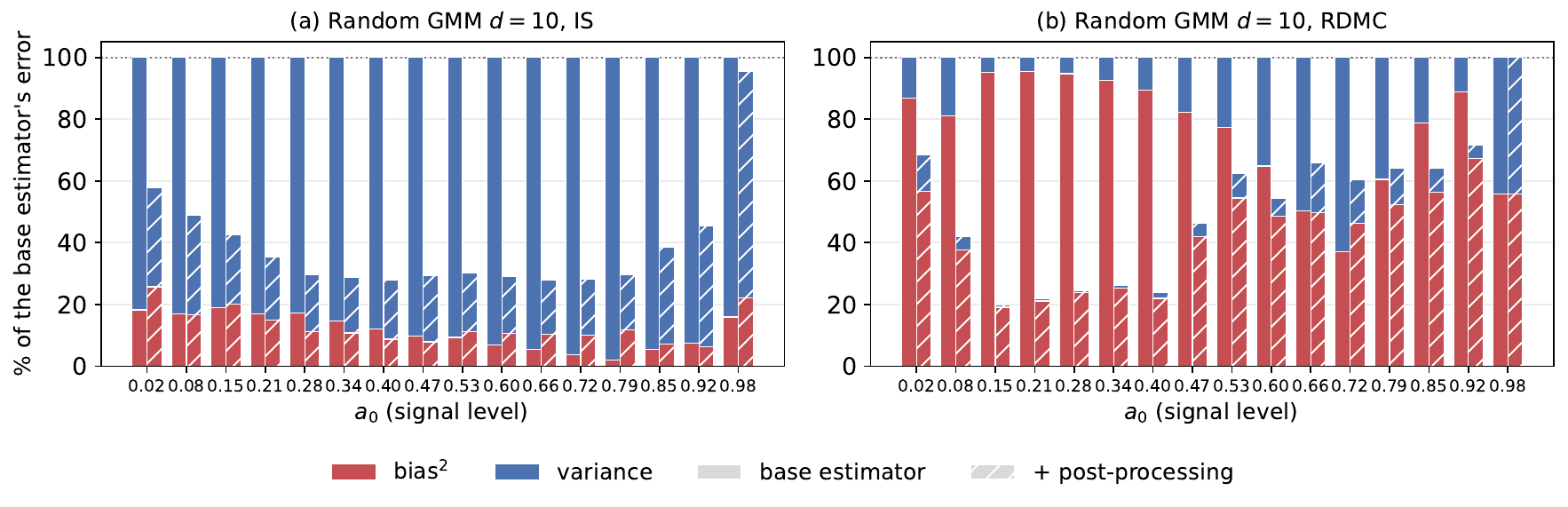}
\caption{Forty-mode mixture at $d=10$, drawn as in Figure~\ref{fig:bias-variance-xgmm}.
The two targets differ in the RDMC panel: here the bias share is large at both ends of
the signal range and dips near $a_0\approx0.6$, so the deepest cuts land in the middle
band rather than at high signal.}
\label{fig:bias-variance-rgmm}
\end{figure}

We see in these figures the dominant method of reducing the $L^{2}$ error is variance reduction rather than bias reduction, and indeed we do see some bias inflation in certain cases. This is to be expected for a shrinkage estimator, the $h$ smoothing vector always has norm less than $1$, so we may interpret the variance term as the weighted average of the on and off-policy variances scaled by the norm of $h$
\begin{align}
    \|h\|^{2} \sum_{i=0}^{N} \frac{h_{i}^{2}}{\|h\|^{2}}\mathcal{V}_{i,i}.
\end{align}

In many cases we see the chosen $h$ vector roughly maintains the bias while drastically reducing the variance with a small norm.

This is however not always the case. We see in Figure \ref{fig:bias-variance-rgmm} the affine post-processing estimator does drastically reduce the bias of the on-policy estimator. This is most likely via the use of the terminal conditions $s(u,0), s(u,1)$ which can be relied on more heavily when the score estimates are structurally incorrect.

\section{Target Densities}
\label{app:targets}
\label{app:geometry-targets}

Write $\Lambda_\kappa=\mathrm{diag}\big(\kappa^{(i-1)/(d-1)}\big)_{i=1}^{d}$ for the diagonal
matrix with eigenvalues log-spaced over $[1,\kappa]$ and $R_\theta$ for the $2\times2$
rotation by $\theta$. Every target except the banana and the shell is a Gaussian mixture
$\pi=\sum_i w_i\,\mathcal{N}(\mu_i,\Sigma_i)$, so the noised density is available in closed form,
\[
p_a(u)=\sum_i w_i\,\mathcal{N}\big(u;\ a\mu_i,\ a^2\Sigma_i+(1-a^2)I\big),
\]
which gives the exact score $s(u,a)$ used to compute $J$, and exact samples for $W_2$.
The banana and the shell have no closed-form $p_a$, so we report only $W_2$ on them.

\paragraph{Ill-conditioned Gaussian.} $\pi=\mathcal{N}(0,\Lambda_\kappa)$, swept over $d$ and $\kappa$.

\paragraph{Banana.} The Haario twisted Gaussian in $d$ dimensions:
$x_1\sim\mathcal{N}(0,s_1^2)$, $x_2\mid x_1\sim\mathcal{N}\big(\gamma(x_1^2-s_1^2),1\big)$, and
$x_3,\dots,x_d\sim\mathcal{N}(0,1)$ i.i.d., with $s_1=3$ and $\gamma=0.1$. The curvature lives in
the $(x_1,x_2)$ plane and the remaining $d-2$ coordinates are standard normal, so the dimension
sweep adds isotropic directions to a fixed twist rather than changing the twist itself. The map
$u_1=x_1/s_1$, $u_2=x_2-\gamma(x_1^2-s_1^2)$, $u_i=x_i$ for $i\geq3$ is an invertible
reparametrisation of an i.i.d.\ standard Gaussian with constant Jacobian determinant $s_1$, so
$\log\pi$ is available in closed form and ground-truth samples are drawn ancestrally.

\paragraph{Orthogonal-covariance mixture (x\_gmm).} In $d=2k$, an equal-weight mixture of $2k$
zero-mean Gaussians whose covariances are block diagonal with $2\times2$ blocks
$R_{\pm45^\circ}\mathrm{diag}(3^2,0.3^2)R_{\pm45^\circ}^{\top}$. The components share a mean,
so they differ only in shape.

\paragraph{Ellipsoidal shell (ell\_shell).} $\pi(x)\propto\exp\!\big(-(\rho(x)-r)^2/2\sigma^2\big)$
with $\rho(x)=\sqrt{x^{\top}\Sigma^{-1}x}$ and $\Sigma=Q\Lambda_\kappa Q^{\top}$, so the mass lies
on an ellipsoid rather than at a point. We use $r=4$, $\sigma=0.3$, $\kappa=50$ and $Q$ a fixed
random rotation. Exact samples come from drawing the Mahalanobis radius by inverse CDF and a
uniform direction on the sphere.

\paragraph{Forty-mode mixture (random\_gmm).}\label{app:random-gmm} $\pi=\frac1{40}\sum_{i=1}^{40}\mathcal{N}(\mu_i,0.05^2I)$
with $\mu_i\sim\mathrm{Unif}([0,2]^d)$, held fixed across seeds. 

\paragraph{Separated anisotropic mixture (sep\_gmm).}
$\pi=0.7\,\mathcal{N}\big(\tfrac23\delta,\ Q\Lambda_\kappa Q^{\top}\big)+0.3\,\mathcal{N}\big(-\tfrac13\delta,\ Q\tilde\Lambda_\kappa Q^{\top}\big)$,
with $Q$ a fixed random rotation, $\tilde\Lambda_\kappa$ equal to $\Lambda_\kappa$ in reverse
order (so the two components have their stiff and loose directions swapped) and $\kappa=10$.
The offset $\delta=\mathrm{sep}\cdot v/\sqrt{v^{\top}\bar\Sigma^{-1}v}$, with $v=Q\mathbf{1}$ and
$\bar\Sigma$ the average of the two covariances, makes $\mathrm{sep}$ the Mahalanobis distance
between the modes at every $d$.

\paragraph{Hypercube-vertex mixture (vertex\_gmm).}\label{app:vertex-gmm}
$\pi=\frac18\sum_{k=1}^{8}\mathcal{N}(\mu_k,\sigma_k^2I)$ with $\mu_k=c\,Qv_k+\tau\bar r\,u$.
Here $v_1,\dots,v_8$ are vertices of $\{\pm1\}^d$ with an even number of $-1$ entries, chosen
by farthest-point sampling from a random start (among $4{,}096$ random such vertices when
$d>20$). $Q$ is a fixed random rotation, and $c$ scales the centers so the closest pair is
$\Delta=12$ apart. The shift by $\tau=0.25$ times the mean radius
$\bar r=\frac18\sum_k\|cQv_k\|$, along a random unit vector $u$, puts the modes at different
distances from the origin. The widths are $\sigma_k=\rho^{s_k/d}$, where $\rho=10^{-2}$ and
$(s_k)$ is a random permutation of $\{0,\frac17,\dots,1\}$, so the volume ratio
$\sigma_{\min}^d/\sigma_{\max}^d=\rho$ is the same at every $d$.

\paragraph{Four-mode anisotropic mixture (zodmc).} The two-dimensional mixture of
\citet{he2024zeroth}, $\pi=\sum_{i=1}^{4}w_i\,\mathcal{N}\big(\tfrac{\mathrm{sep}}{11}m_i,\Sigma_i\big)$
with $w=(0.1,0.2,0.3,0.4)$, $m_i\in\{(0,0),(0,11),(9,9),(11,0)\}$ and
$\Sigma_1=\left(\begin{smallmatrix}1&0.5\\0.5&1\end{smallmatrix}\right)$,
$\Sigma_2=\left(\begin{smallmatrix}0.3&-0.2\\-0.2&0.3\end{smallmatrix}\right)$,
$\Sigma_3=\left(\begin{smallmatrix}1&0.3\\0.3&1\end{smallmatrix}\right)$,
$\Sigma_4=\left(\begin{smallmatrix}1.2&-1\\-1&1.2\end{smallmatrix}\right)$.
$\mathrm{sep}=11$ recovers their original means.

\subsection{Does the covariance have to be estimated?}
\label{app:sigma-ablation}

In our experiments, the weights $h,g$ solve the minimization problem 

\begin{align}
(h^{*},g^{*}) = \operatorname*{arg\,min}_{h\in\mathbb{R}^{N+1},\,g\in\mathbb{R}^{2}}\ \sum_{k\notin\mathcal{K}}\Omega_{k,k}\,\delta_{k}(h,g)^{2} + \lambda\,h^{\top}\Sigma h \quad \text{s.t.}\quad \delta_{k}(h,g) = 0,\ k\in\mathcal{K}.
\end{align}

We use $\Sigma$ as an empirically estimated covariance matrix of the score estimator $\hat{s}$ at the different $\{a_{i}\}_{i=0}^{N}$ design levels. This requires multiple calls to the score estimator to estimate the covariance matrix.

We ablate it by replacing $\hat\Sigma$ with
three progressively coarser surrogates, holding nodes, degree, $\lambda$, misfit weights,
endpoint anchors and oracle budget fixed: $\mathrm{diag}(\hat\Sigma)$, which keeps the
per-level variances and drops the cross-level correlations; $c\,I$ with
$c=\frac1n\sum_i\hat\Sigma_{ii}$, which keeps only the overall scale of the estimated
covariance and discards its shape; and $I$, which measures nothing. The budget is unchanged
throughout, since the replicate groups already feed the reported estimate. Note that $c\,I$
and $I$ are not the same rule: $\mu$ tracks $\operatorname{tr}\Sigma$, so the solve is
invariant to rescaling $\Sigma$ everywhere except the exact endpoint rows, whose variance is
an absolute quantity that does not rescale. Forcing the interior to unit scale therefore
multiplies the trust placed in the endpoint anchors by $1/c$, which is inert when $c\approx1$
and severe on stiff targets.

Figure~\ref{fig:sigma-ablation} reports $W_2$ against dimension for the four models; the
comparisons quoted below are per-seed ratios to $\hat\Sigma$. Two things stand out. 
The correlations are worth nothing: $\mathrm{diag}(\hat\Sigma)$
sits on unity in every panel at every dimension, so estimating $n(n+1)/2$ entries buys nothing
over estimating the $n$ diagonal terms. This is consistent with Appendix~\ref{app:bias-variance}: 
$\Sigma$ enters only the variance half of the error, and that half shrinks as the base estimator's importance
weights degenerate, so an inverse-variance tilt optimises a diminishing share of the error
while distorting $h$ away from the allocation that reproduces the basis.

\begin{figure}[H]
\centering
\includegraphics[width=\textwidth]{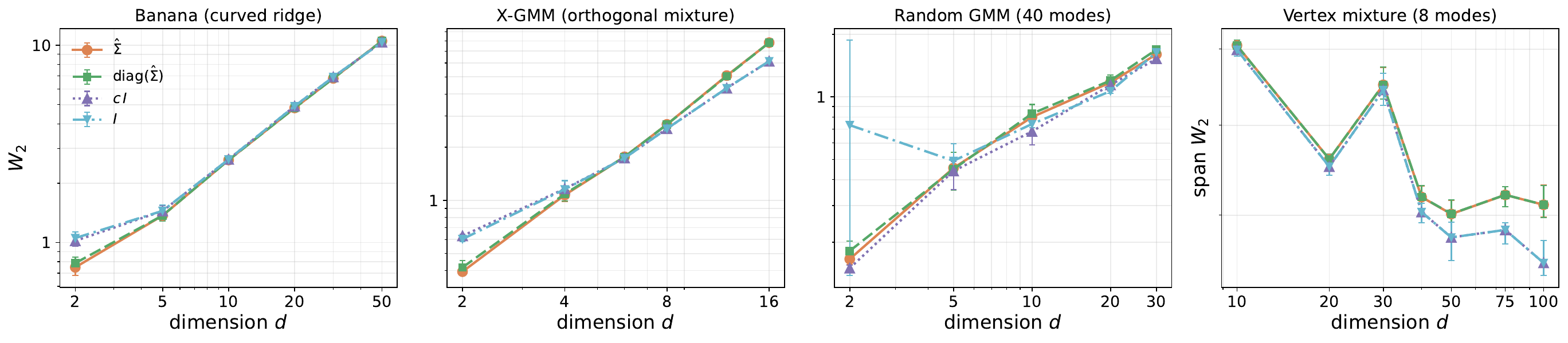}
\vspace{-\baselineskip}
\caption{Covariance ablation: $W_2$ against dimension for the four choices of $\Sigma$, at a
matched oracle budget. Bars are [min, max] over 3 seeds and the vertex panel reports span
$W_2$. $\hat\Sigma$ and $\mathrm{diag}(\hat\Sigma)$ coincide to within about $1\%$, as do
$c\,I$ and $I$, so each model carries its own dash pattern and marker.}
\label{fig:sigma-ablation}
\end{figure}

\section{Score Estimation Subroutine}
\label{app:algorithm}

Algorithm~\ref{alg:score} gives the post-processed score $\bar s(u,a_0)$ that replaces the
base estimator in the reverse process of \eqref{eq:em}, written in the form of
the constrained quadratic program~\eqref{eq:practical_qp} of Section~\ref{sec:nonoracle},
so that line~6 returns the weights $(h^{*},g^{*})$ of Theorem~\ref{ridge_adjoint} and
line~7 is the affine combination~\eqref{eq:affine_combination}. Each design level is
queried $N_{\mathrm{batch}}$ times independently, each query being one full estimate from
$n$ Monte Carlo draws, for a budget of $N_{\mathrm{batch}}(N{+}1)n$ draws per score
($2\times6\times5{,}000=60{,}000$ in our experiments). The spread between replicates at the
same level gives the variances $V_i$, and their mean is the data the weights are applied to.

\begin{algorithm}[H]
\caption{Post-processed score $\bar s(u,a_0)$. The values used in the experiments are
collected in Table~\ref{tab:protocol}.}
\label{alg:score}
\begin{algorithmic}[1]
\Require location $u\in\mathbb{R}^d$; signal level $a_0$; base estimator $\hat s(u,a;n)$;
grid of $N{+}1$ levels uniform on $[\underline a,\overline a]$; Legendre basis
  $\phi_0,\dots,\phi_K$; replicates $N_{\mathrm{batch}}$; ridge weight $\lambda$.
\State Replace the grid level nearest $a_0$ by $a_0$, giving design levels $(a_0,\dots,a_N)$
\State $\hat s_{r,i}\gets\hat s(u,a_i;n)$ for $i=0,\dots,N$ and $r=1,\dots,N_{\mathrm{batch}}$
  \Comment{independent calls}
\State $\bar s_i\gets\frac{1}{N_{\mathrm{batch}}}\sum_{r}\hat s_{r,i}$
\State $V_i\gets\frac{1}{(N_{\mathrm{batch}}-1)\,d}\sum_{r}\lVert\hat s_{r,i}-\bar s_i\rVert^2$
  \Comment{variance of the base estimator at level $a_i$}
\State $v_0\gets-u$,\quad $v_1\gets\nabla\log\pi(u)$
  \Comment{exact scores at $a=0,1$}
\State $\displaystyle (h,g)\gets\operatorname*{arg\,min}_{h\in\mathbb{R}^{N+1},\,g\in\mathbb{R}^{2}}
  \sum_{k=1}^{K}\frac{1}{k^{3}}\Big(\sum_{i=0}^{N}h_i\phi_k(a_i)+g_0\phi_k(0)+g_1\phi_k(1)-\phi_k(a_0)\Big)^{2}
  +\lambda\sum_{i=0}^{N}V_ih_i^{2}$ \Comment{problem~\eqref{eq:practical_qp}}
\Statex \hfill s.t.\ $\sum_{i=0}^{N} h_{i}\phi_{k}(a_{i}) + g_{0}\phi_{k}(0) + g_{1}\phi_{k}(1) = \phi_{k}(a_{0}), \qquad k = 0$
\State \Return $\bar s(u,a_0)=\sum_{i=0}^{N}h_i\,\bar s_i+g_0v_0+g_1v_1$ \Comment{eq.~\eqref{eq:affine_combination}}
\end{algorithmic}
\end{algorithm}

\paragraph{Implementation details.} Here we present some minor implementation modifications that
we found helpful in practice.

\emph{The penalty is placed on the misfit rather than on the variance.} Writing
$\tilde h=(h,g)$, $\tilde A=[\phi(a_0),\dots,\phi(a_N),\phi(0),\phi(1)]$ with
$\phi(a)=(\phi_0(a),\dots,\phi_K(a))^{\top}$, $b=\phi(a_0)$ and
$W=\mathrm{diag}(w_0,1^{-3},2^{-3},\dots,K^{-3})$, the implementation solves
\begin{align}
\tilde h^{*}=\operatorname*{arg\,min}_{\tilde h}\ \tilde h^{\top}\tilde\Sigma\tilde h
+\mu\,\lVert W^{1/2}(\tilde A\tilde h-b)\rVert^{2}
=\mu\big(\tilde\Sigma+\mu\tilde A^{\top}W\tilde A\big)^{-1}\tilde A^{\top}Wb,
\qquad \mu=\lambda\,\frac{\operatorname{tr}\Sigma}{N+1}.
\label{eq:impl_solve}
\end{align}
Dividing through by $\mu$ shows this is the objective of line~6, i.e.\ the adjoint
problem~\eqref{adjoint_prob} with $\Omega_{k,k}=k^{-3}$ and $\mathcal{K}=\{0\}$, with ridge
weight $1/\mu$ and the degree-$0$ constraint replaced by a large weight $w_0=10^{4}$ on its
misfit; the closed form is then the soft-constraint counterpart of $(h^{*},g^{*})$ in
\eqref{eq:ridge_closed}. We use the large weight rather than an exact constraint so that the
weights are given by a single linear solve, which stays well conditioned even though the
endpoint rows of $\tilde\Sigma$ have zero variance.
The factor $\operatorname{tr}(\Sigma)/(N+1)$ makes the trade-off between the two terms
invariant to the overall scale of the base estimator's variance, which differs by orders of
magnitude across targets and along a single reverse trajectory. Without it a single
$\lambda$ would be nearly unbiased on one target and heavily shrunk on another. By doing so
we are able to off load the tuning of $\lambda$ to the overall scale of the base estimator's variance.
We report this $\lambda$ in Table~\ref{tab:protocol} and Appendix~\ref{app:hyperparam}.

\emph{The endpoint constraints are soft.} Theorem~\ref{ridge_adjoint} treats the exact
scores $s(u,0)=-u$ and $s(u,1)=\nabla\log\pi(u)$ as hard constraints, equivalently as
observations of zero variance. The implementation gives them the small variances
$\epsilon_e=\epsilon\,\lVert v_e\rVert^{2}/d$, $e\in\{0,1\}$, with $\epsilon=10^{-4}$, so
that $\tilde\Sigma=\mathrm{blkdiag}(\Sigma,\epsilon_0,\epsilon_1)$ in
\eqref{eq:impl_solve}; $\epsilon\to0$ recovers the hard constraint. This is done for numerical stability
and does not affect the final solution in practice. 

\section{Hyperparameter Selection}
\label{app:hyperparam}
Our algorithm has several hyperparameters that need to be chosen. Specifically, we have to choose the hyperparameters for the base estimator, which for importance sampling and rejection sampling it is simply a constant $M$ saying how any Monte Carlo samples to use. Then there is the size of the grid $K$, and the degree of the polynomial $K_{deg}$, as well as the number of batches $N$ to estimate the covariance matrix, and finally the hyperparameter $\lambda$ for the regression objective.

Because $N$ the number of batches needed to estimate the covariance matrix for the regression problem, before tuning anything downstream we first check how well $\Sigma$ itself is estimated, since every other hyperparameter will depend on it. 

\subsection{The reverse grid}
\label{app:grid}

We place the $20$ levels of \eqref{eq:em} by the step-size rule
\begin{align}
    \tau_{k} = c\,\min\{t_{k},\,1\}, \qquad t_{k} = -\log a_{k},
    \label{eq:grid-rule}
\end{align}
in the diffusion time $t$ of the equivalent variance-preserving process, so that $t_{k}$
is the distance to the data end $a=1$: steps are uniform in $t$ while the noise
dominates, and shrink in proportion to $t_{k}$ once $t_{k}<1$. This is the
exponentially decreasing schedule of
\citet{chen2023improved,he2024zeroth,benton2024nearly}, where it is what gives the
optimal dependence on the early-stopping time.

We use $c=0.85$, which places the grid close to uniform in the log signal-to-noise ratio
$\lambda=\log\!\big(a^{2}/(1-a^{2})\big)$ the spacing used by
\citet{kingma2021variational} and \citet{lu2022dpm}. The levels then run from
$a_{0}=6.7\times10^{-3}$ to $a_{19}=0.999$ with $h_{k}/a_{k}\leq0.68$ throughout. We do not tune $c$ or the grid any further.

\subsection{The configuration used in the experiments}
\label{app:hyperparam-config}

Table~\ref{tab:protocol} collects the set of hyperparameters ran/sweep for different methods across all of our experiments. The compute budget is stated in target-gradient evaluations and is matched
across samplers: each particle costs $60{,}000 \times 20 = 1.2$M evaluations, whether
they are spent as importance-sampling draws, as RDMC inner-chain steps, or as moves of a classical chain. 
\begin{table}[h]
\centering
\caption{Choice of hyperparameters for different methods}
\label{tab:protocol}
\begin{tabular}{lll}
\toprule
 & quantity & value \\
\midrule
\multirow{4}{*}{Reverse process}
 & discretization steps & $20$ \\
 & forward schedule & variance preserving linear schedule, $t\in[10^{-3}, 0.9933]$ \\
 & seeds per cell & $3$ \\
\midrule
\multirow{3}{*}{Compute budget}
 & oracle calls per score query & $60{,}000$ \\
 & calls per particle & $60{,}000 \times 20 = 1.2$M \\
\midrule
\multirow{3}{*}{Post-processing (ours)}
 & base-estimator calls per step & $12$ \\
 & draws per call & $5{,}000$ \\
 & ridge weight $\lambda$ & $1$ (see Appendix~\ref{app:hyperparam}) \\
\midrule
\multirow{4}{*}{RDMC}
 & inner chain length $K$ & $1000$ \\
 & inner step $h$ (swept) & $\{10^{-5},\ 3\times10^{-5},\ 10^{-4},\ 2\times10^{-4},$ \\
 & & $\phantom{\{}5\times10^{-4},\ 10^{-3},\ 2\times10^{-3},\ 3\times10^{-3},$ \\
 & & $\phantom{\{}5\times10^{-3},\ 7\times10^{-3},\ 0.01\}$ \\
 & inner step $h$ & $7\times10^{-3}$ ($2\times10^{-4}$ on banana) \\
\midrule
\multirow{5}{*}{Classical chains}
 & gradient steps per chain & $1.2$M (matched) \\
 & PT step size & $1/L$ \\
 & LMC step size (swept) & $\{3\times10^{-5},\ 10^{-4},\ 3\times10^{-4},\ 10^{-3},$ \\
 & & $\phantom{\{}3\times10^{-3},\ 0.01,\ 0.03,\ 0.05,$ \\
 & & $\phantom{\{}0.1,\ 0.5,\ 1\}$ \\
\bottomrule
\end{tabular}
\end{table}

Every number reported for the ill-conditioned Gaussian uses one configuration, unchanged
across all $16$ cells of the $d\times\kappa$ grid:

\begin{itemize}
\item \textbf{Design levels.} $N+1=6$ nodes on a uniform grid over $b\in[0.05,0.95]$, with
the query level $a_0$ substituted for whichever grid point is nearest to it. Substituting
rather than appending matters: appending lets two design levels coincide when $a_0$ falls
near a grid point, which makes two rows of $\hat\Sigma$ near-identical and the solve
ill-conditioned.
\item \textbf{Degree.} $K_{\mathrm{deg}}=5$, the largest degree the node count supports
($A$ is $(K_{\mathrm{deg}}{+}1)\times(N{+}1)$, so $K_{\mathrm{deg}}\le N$ is required for
$Ah=b$ to be solvable).
\item \textbf{Covariance batches.} $N_{\mathrm{batch}}=2$ replicate groups, pooled across
the particle batch, and the mean of those groups is reused as the data $h$ is applied to
rather than drawing a further independent group. Reuse costs a plug-in bias of order
$1/N_{\mathrm{batch}}$ but cuts the reported estimate's variance by the same factor and
saves $N{+}1$ calls; measured, it was worth roughly $20\%$ in $W_2$.
\item \textbf{Misfit weights.} $w_k\propto k^{-3}$, with the degree-$0$ constraint enforced
exactly. By the primal-dual identity of Section~\ref{sec:nonoracle} this is exactly a $k^3$
ridge on the fitted coefficients, expressed in the dual variables.
\item \textbf{Endpoint anchors.} The two exactly known scores $s(u,0)=-u$ and
$s(u,1)=\nabla\log\pi(u)$ enter as extra design nodes with variance
$\epsilon=10^{-4}\,\overline{s_{\mathrm{end}}^2}$. They cost no Monte Carlo, so the budget
is unchanged. The $\epsilon\to0$ limit -- treating them as exact, zero-variance
constraints -- diverges on targets with a stiff score, because $\nabla\log\pi$ is unbounded
while the base estimate is implicitly bounded by the proposal geometry; scaling $\epsilon$
by the anchor's own magnitude penalizes its \emph{contribution} rather than its weight, and
is inert where the anchor is well-scaled.
\item \textbf{Regression penalty.} $\lambda=1$, fixed.
\end{itemize}

Together these cost $N_{\mathrm{batch}}\,(N{+}1)=12$ base-estimator calls per reverse step,
against \texttt{raw}'s one, and the per-call Monte Carlo budget is divided by $12$ so the
total spend per step matches.

\paragraph{On not tuning $\lambda$ per target.} $\lambda$ is the one hyperparameter whose
optimum visibly moves across the grid, and it would be easy to report the best value in
each cell. We do not, because that is an oracle choice a deployed sampler does not get. The
cost of fixing it is small and worth stating: selecting the best $\lambda\in\{1,10\}$ per
cell improves the geometric-mean gain over \texttt{raw} from $1.94\times$ to $2.09\times$
across the grid, and by $0$--$14\%$ in each cell where the method wins. We report the
untuned $\lambda=1$ number throughout.

\paragraph{A configuration we exclude.} Large $\lambda$ ($\ge100$) is unstable: across the
grid it diverges on a subset of seeds in five separate cells, with a seed standard deviation
comparable to or larger than the mean, while $\lambda\in\{1,10\}$ is stable in the same
cells. Large $\lambda$ drives the fit toward exact mean preservation and away from variance
damping, which is the same regime in which the $\epsilon\to0$ endpoint limit fails. We
exclude it rather than report a number produced by a diverging run.

\subsection{Convergence of the Estimated Score Covariance $\hat\Sigma$}
\label{app:sigma-convergence}

\paragraph{Setup.} Fix a target from Appendix~\ref{app:targets} and a set of $M \in \{10, 50, 150\}$ points
$x_1,\dots,x_M$ drawn on-trajectory at a reference noise level $b$ (i.e. $x_i = a_b x_0^{(i)} + s_b z^{(i)}$ for $x_0^{(i)}
\sim \pi$). Given $N$ independent replicate score evaluations at each of the $N_\ell = 5$ evaluation levels
$b\in\{0.1,0.3,0.5,0.7,0.9\}$, pooled over the $M$ points and the coordinate dimension, we form the empirical covariance
\[
\hat\Sigma_N[i,j] \;=\; \frac{1}{(N-1)Md}\sum_{g=1}^{N}\sum_{m=1}^{M}\sum_{k=1}^{d} \big(\hat s_g(x_m,b_i) - \bar s(x_m,b_i)\big)\big(\hat s_g(x_m,b_j) - \bar s(x_m,b_j)\big),
\]
the same weighted-covariance quantity used internally by the regression. We draw the $N$ replicate groups such that
the sequence $\hat\Sigma_2,\hat\Sigma_5,\dots,\hat\Sigma_{500}$ is a sequence converging to the same value.

\paragraph{Metric.} We compute the absolute Frobenius error $\|\hat\Sigma_N - \Sigma\|_F$ against a reference 
$\Sigma := \hat\Sigma_{500}$, as a function of $N$. We repeat this procedure 5 times, and plot the median along with the range $[\min,\max]$ range of the error.
\paragraph{Comparison.} We compare two base estimators: importance sampling with $1000$ Monte Carlo samples \emph{per time} and rejection sampling with $3000$ samples \emph{per time}, this larger per-point number was chosen to off-set the low acceptance rate of the rejection sampler.
We vary the pooling size $M \in \{10, 50, 150\}$ to see how using more groups changes the convergence picture as a function of the group size $N$.

\begin{figure}[H]
\centering
\includegraphics[width=\textwidth]{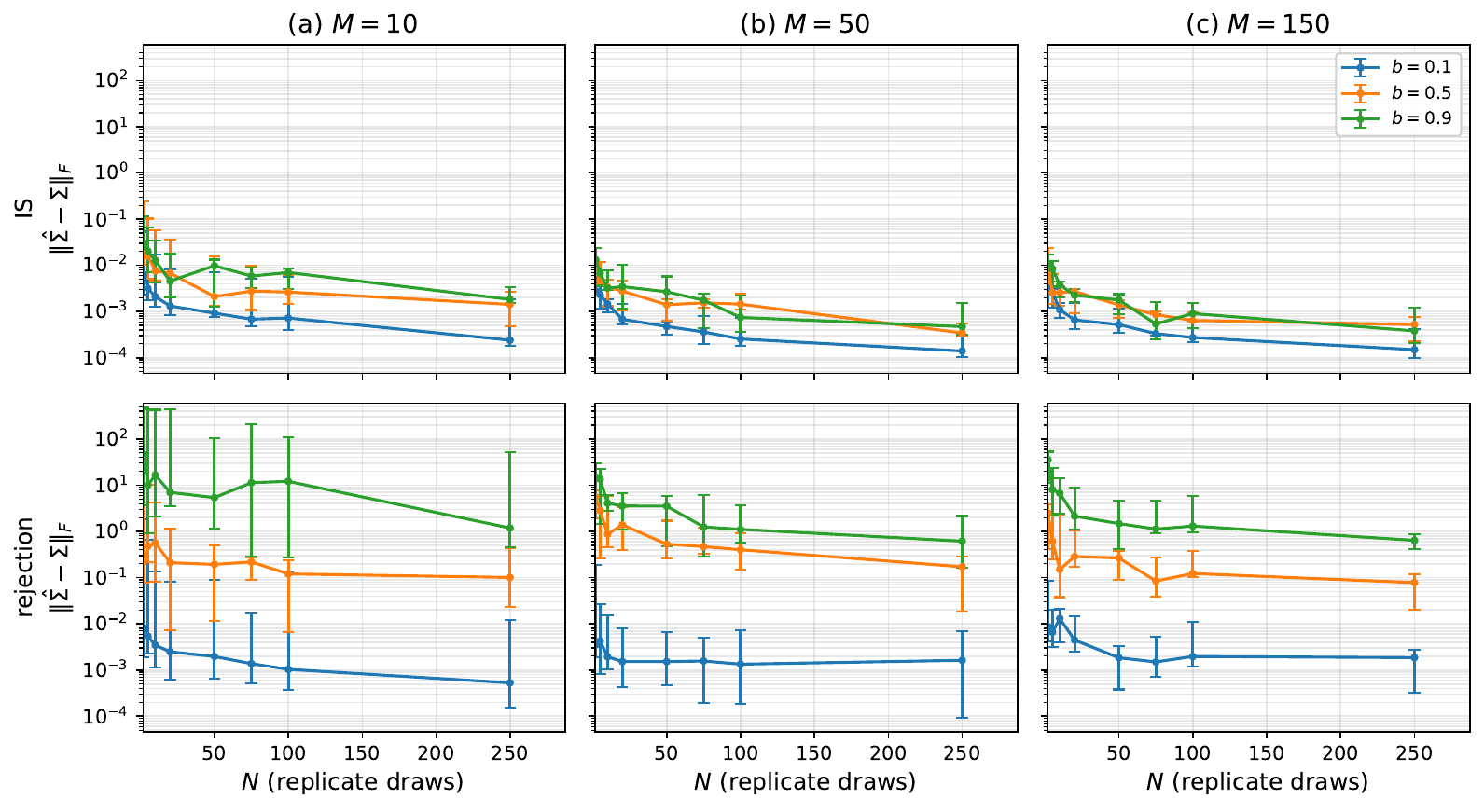}
\caption{$\hat\Sigma$-convergence, ill-conditioned Gaussian ($d=2$, $\kappa=1000$). Rows: IS (top) vs.\ rejection sampling (bottom). Columns: pooling size (a) $M=10$, (b) $M=50$, (c) $M=150$. Each curve is one reference noise level $b$; markers are the median and error bars the full range over 5 independent trials.}
\label{fig:sigma-convergence-logconcave}
\end{figure}

\begin{figure}[H]
\centering
\includegraphics[width=\textwidth]{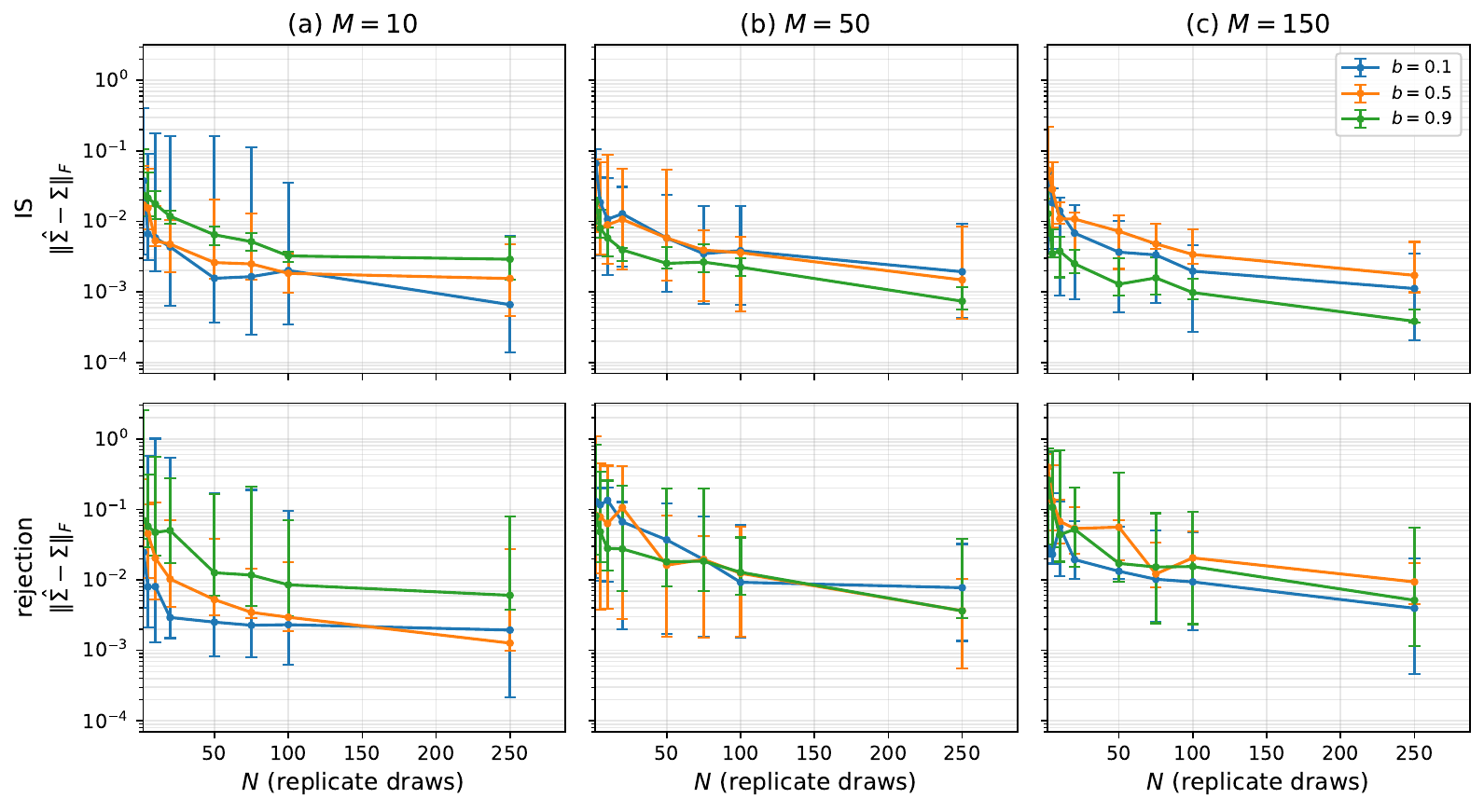}
\caption{$\hat\Sigma$-convergence, orthogonal-covariance mixture (x\_gmm, $d=2$). Same layout as Figure~\ref{fig:sigma-convergence-logconcave}.}
\label{fig:sigma-convergence-xgmm}
\end{figure}

\begin{figure}[H]
\centering
\includegraphics[width=\textwidth]{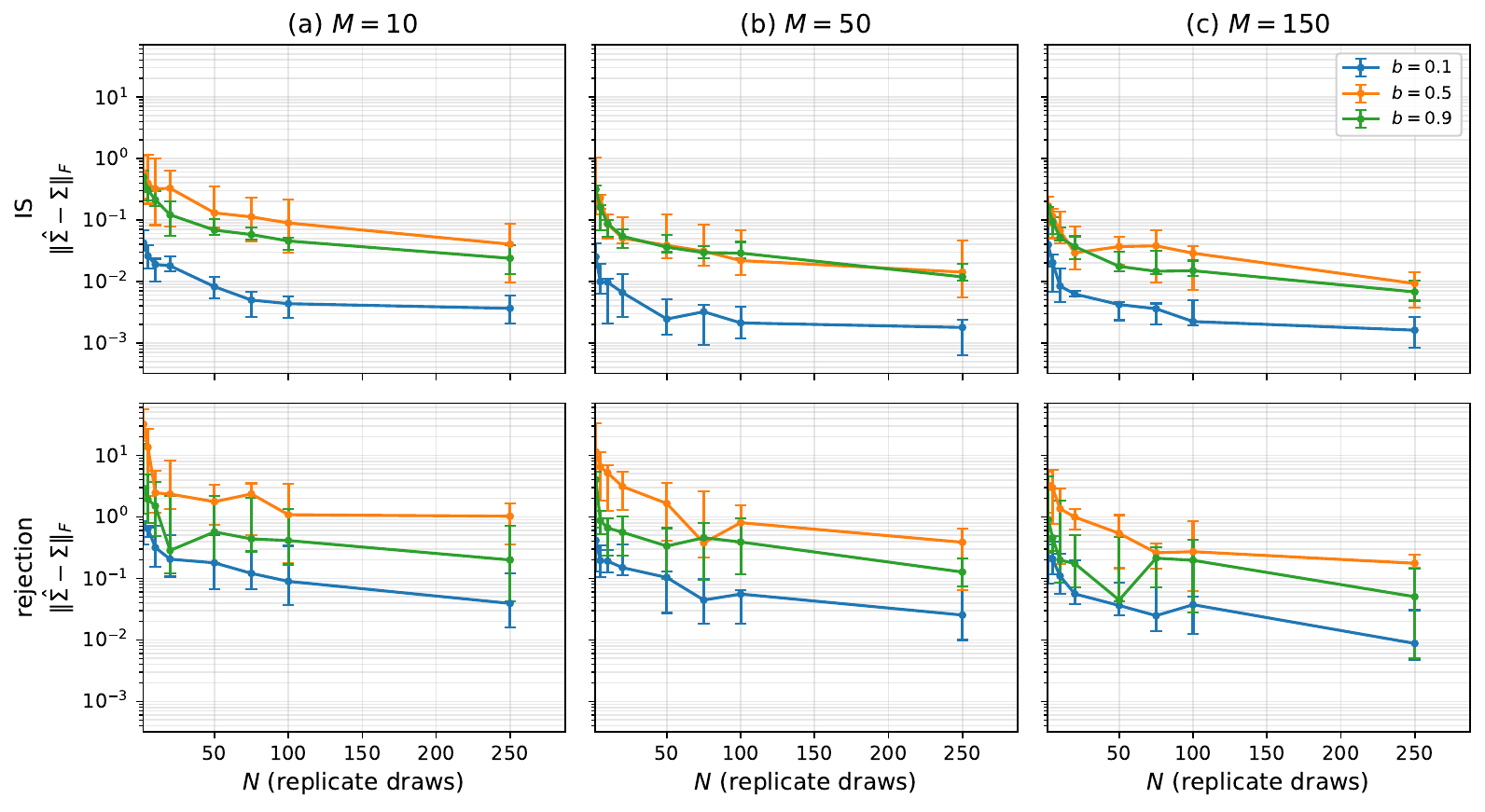}
\caption{$\hat\Sigma$-convergence, four-mode anisotropic mixture (zodmc, $\mathrm{sep}=11$). Same layout as Figure~\ref{fig:sigma-convergence-logconcave}.}
\label{fig:sigma-convergence-zodmc}
\end{figure}

\paragraph{Findings.} Across all three targets, importance sampling's $\hat\Sigma$ converges quickly: by $N\approx 150$
replicate draws its absolute Frobenius error to the $N=500$ reference is at or below $10^{-2}$ for every $M$ and $b$ we tried.
Rejection sampling's $\hat\Sigma$ is noisier and slower converging in every target, which is expected considering the acceptance rate. 
We find that the group size $M$ has a smaller effect in this regime; on the other hand, the choice of $N$ is more important. 
A precise $\hat\Sigma$ takes tens of replicates. A crude estimate of $\Sigma$, however, is enough to set the
weights, consistent with the observation in Appendix~\ref{app:sigma-ablation} that $\Sigma = I$ performs comparably well. We
therefore estimate it with only $N_{\mathrm{batch}}=2$ replicates in our experiments, which keeps each replicate's Monte Carlo sample large
and so limits the bias of self-normalized importance sampling. When more compute is available, more replicates can be used.

\subsection{Optimal weight allocation on the geometry targets}
\label{app:geometry-allocation}

Figure~\ref{fig:geometry-j} reports how much of the recoverable error the constrained
estimator captures. Figure~\ref{fig:geometry-allocation} shows \emph{how} it does so, and
what the unrestricted optimum would do instead. Each column is a query level $a_0$, each
row a design level, coloured by $|h_i|$ normalized within the column -- so a column reads
as the recipe for that query's estimate.

\begin{figure}[H]
\centering
\includegraphics[width=\textwidth]{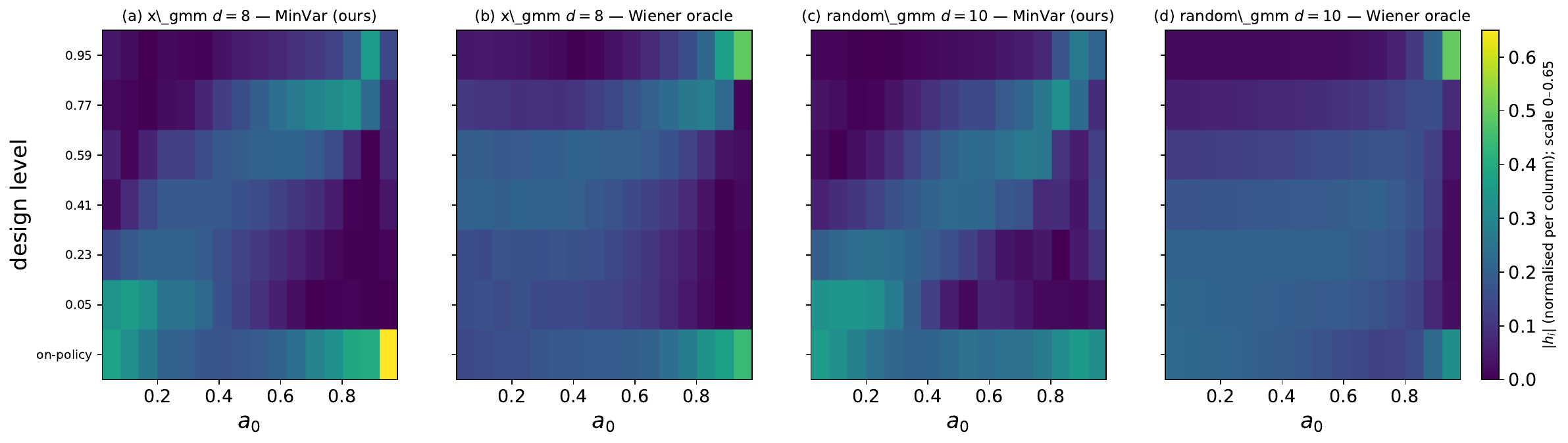}
\caption{Weight allocation on the orthogonal mixture and the $40$-mode mixture, at the
high-dimensional cell of each. Panels alternate ours and the Wiener oracle -- the minimum over all
affine post-processors -- for each target. Colour scale is shared across all four panels.}
\label{fig:geometry-allocation}
\end{figure}

Two features recur from the ill-conditioned Gaussian. The oracle spreads its weight more
evenly across design levels than the constrained rule does, which concentrates on the
levels nearest the query and leans increasingly on the on-policy estimate as $a_0$ grows.
And at small $a_0$ -- where the recoverable gap is largest -- the oracle draws on distant
levels that mean preservation forbids, which is where the residual $12$--$24\%$ of
unrecovered error sits. The two allocations agree most closely in the mid-range, which is
also where the $J$ curves of Figure~\ref{fig:geometry-j} are closest together.

\begin{figure}[H]
\centering
\includegraphics[width=\textwidth]{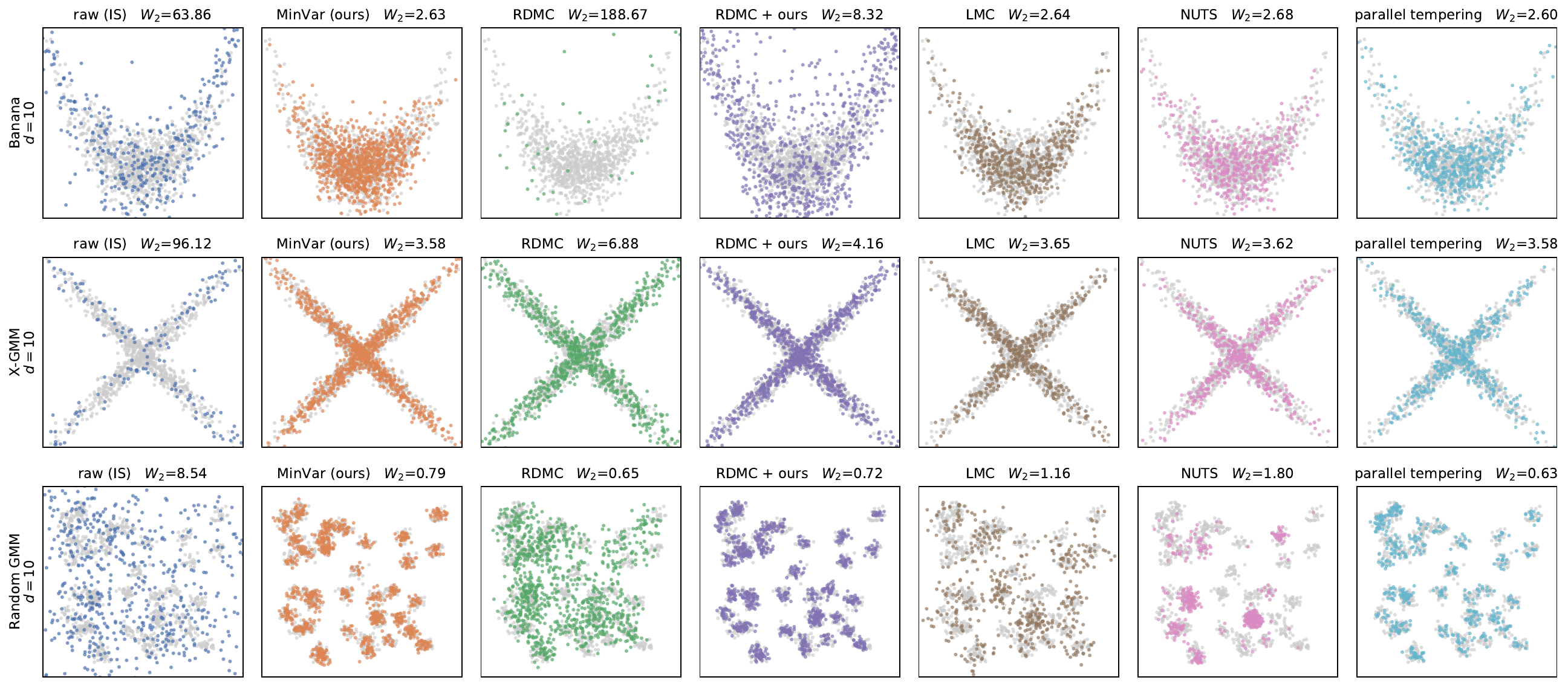}
\caption{Samples at $d=10$ on the banana, the orthogonal mixture and the $40$-mode mixture, first two coordinates, with the target in gray behind each panel. Equal compute, one seed; $W_2$ is quoted above each panel. Figure~\ref{fig:vertex-dim}(d) shows the $40$-mode mixture at $d=30$.}
\label{fig:geometry-scatter-all}
\end{figure}

\end{document}